\documentclass[a4paper,twocolumn,superscriptaddress,accepted=2024-12-17]{quantumarticle}
\pdfoutput=1
\usepackage[T1]{fontenc}
\usepackage[utf8]{inputenc}
\usepackage{graphicx}
\graphicspath{{Figs}}
\usepackage{float}
\usepackage[plainpages=false,pdfpagelabels,hypertexnames=false]{hyperref}
\usepackage[usenames,dvipsnames,table]{xcolor}
\usepackage{braket}
\usepackage{amsmath,amssymb}
\usepackage{bbm}
\usepackage{physics}
\usepackage{amsthm}
\usepackage{bm}
\usepackage[square,numbers]{natbib}
 
\usepackage{ulem}
\usepackage{tikz}
\usepackage{pgfplots}
\usetikzlibrary{arrows}
\usetikzlibrary{calc,shapes.geometric}
\usetikzlibrary{decorations.pathmorphing}

\usepackage{diagbox}

\makeatletter
\newcommand{\pushright}[1]{\ifmeasuring@#1\else\omit\hfill$\displaystyle#1$\fi\ignorespaces}
\makeatother

\newtheorem{lem}{Lemma}

\newtheorem{ub}{Upper bound}

\newcommand{\JM}{\textit{JM}}

\theoremstyle{definition}

\usepackage[caption=false]{subfig}

\tikzset{
  antenna/.style={
    isosceles triangle,
    fill=black,
    minimum width=0.2cm,
    inner sep=0pt,
    shape border rotate=90
  }
}

\pgfplotsset{
  myplotset/.style={
    tick style = {thick},
    grid = major,
    major grid style = {gray!25},
  }
}

\def\jmathphys{J.\ Math.\ Phys.}%
\def\josab{J.\ Opt.\ Soc.\ Am.\ B}%
\def\jphysa{J.\ Phys.\ A}%

\def\njp{New J.\ Phys.}%
\def\pra{Phys.\ Rev.\ A}%
\def\prl{Phys.\ Rev.\ Lett.}%
\def\prsa{Proc.\ R.\ Soc.\ Lond.\ A}%

\begin{document}

\title{Joint-measurability and quantum communication with untrusted devices}

\author{Michele Masini}
\email{michele.masini@ulb.be}
\affiliation{Laboratoire d'Information Quantique, Universit\'e libre de Bruxelles (ULB), Belgium}

\author{Marie Ioannou}
\affiliation{Department of Applied Physics, University of Geneva, Switzerland}

\author{Nicolas Brunner}
\email{Nicolas.Brunner@unige.ch}
\affiliation{Department of Applied Physics, University of Geneva, Switzerland}

\author{Stefano Pironio}
\email{stefano.pironio@ulb.be}
\affiliation{Laboratoire d'Information Quantique, Universit\'e libre de Bruxelles (ULB), Belgium}

\author{Pavel Sekatski}
\email{pavel.sekatski@gmail.com}
\affiliation{Department of Applied Physics, University of Geneva, Switzerland}

\date{21 March 2024}

\begin{abstract}
Photon loss represents a major challenge for the implementation of quantum communication protocols with untrusted devices, e.g. in the device-independent (DI) or semi-DI approaches. Determining critical loss thresholds is usually done in case-by-case studies. In the present work, we develop a general framework for characterizing the admissible levels of loss and noise in a wide range of scenarios and protocols with untrusted measurement devices. In particular, we present general bounds that apply to prepare-and-measure protocols for the semi-DI approach, as well as to Bell tests for DI protocols. A key step in our work is to establish a general connection between quantum protocols with untrusted measurement devices and the fundamental notions of channel extendibility and joint-measurability, which capture essential aspects of the communication and measurement of quantum information. In particular, this leads us to introduce the notion of partial joint-measurability, which naturally arises within quantum cryptography.
\end{abstract}

\maketitle

\section{Introduction}
Photon loss, which results from unavoidable absorption and scattering in optical channels, as well as limited detector efficiency, represents one of the key challenges for the implementation of long-distance quantum communication. This issue becomes even more critical in experiments and applications where quantum states are transmitted through untrusted channels and measured using untrusted apparatuses. This is common in quantum cryptography, most notably within the device-independent (DI) \cite{Acin2007,diqkd_review} or semi-device-independent (SDI) \cite{pawlowski2011semi,woodhead2015secrecy} approaches, but also in fundamental tests of quantum physics such as quantum Bell nonlocality \cite{Bell,non-locality} and steering \cite{wiseman2007,steering}. 

In these situations, certifying quantum properties such as entanglement or randomness, and achieving quantum advantages, such as information-theoretic security, becomes impossible when photon losses exceed a certain threshold. The underlying reason is that the untrusted channels and measurement devices could potentially implement processes in which losses are not merely passive phenomena, but are influenced by the choice of the experimenter's measurement settings in a way that could, for instance, allow for blinding attacks \cite{lydersen2010hacking,gerhardt2011full,jain2011device,bugge2014laser}. In such attacks, the eavesdropper might selectively and remotely cause detection events based on whether their own measurement matches that of the receiver. More generally, the most famous example of this is arguably the so-called ``detection loophole'' \cite{eberhard_efficiency}, that plagued experimental tests of quantum nonlocality for decades, before experiments could finally close it, and, eventually led to the celebrated loophole-free Bell experiments. 

The level of admissible losses is generally highly dependent on the specific quantum communication protocol that is considered. 
Moreover, this tolerance also depends on the detailed aspects of the experimental setup, including the specific states that are intended to be transmitted through the channel and the expected measurements to be performed on them. An intense effort has been devoted to characterizing critical loss thresholds. This started in the context of Bell experiments (see e.g. \cite{eberhard_efficiency,massar_violation_2003,Vertesi2010,Branciard2011} and \cite{Larsson2014,non-locality} for reviews) as well as in steering test (see e.g. \cite{Wittmann2012,Bennet2012,Srivastav2022}). In turn, critical efficiencies for DI and SDI protocols have also been discussed (see e.g. \cite{acin2016necessary,ioannou2022receiver}). While all these works have played a significant role towards the first loophole-free Bell tests \cite{Hensen2015,Shalm2015,Giustina2015} and proof-of-principle implementations of DI protocols \cite{Pironio2010,Nadlinger2022,Zhang2022}, it is fair to say that our understanding is currently limited to few specific cases. 

In this paper, we develop a general framework for characterizing critical losses in quantum communication setups. This allows us to establish lower-bounds on admissible losses that are applicable to a wide range of experiments and protocols involving untrusted measurement devices. The key point of our work is to establish a connection with the notion of ``joint-measurability'', which captures the incompatibility of quantum measurements \cite{Heinosaari2016,Busch2016,uola2023}.

Our starting point is to introduce the concept of a ``channel-measurement unit'' (CMU), which represents a fundamental component of quantum communication protocols with untrusted devices. Specifically, the CMU consists of an untrusted communication channel $C$ followed by an untrusted measuring device $M$, as depicted in Fig.~\ref{fig:CMU}. The CMU serves as the basic building block for a broad range of setups. The simplest examples are SDI prepare-and-measure scenarios \cite{Gallego2010,Wang2019}, as in Fig.~\ref{fig:exp}.a), and bipartite Bell experiment where a pair of entangled photons are sent to two separated CMUs (see Fig.~\ref{fig:exp}.b.). These setups serve as a basis for SDI QKD protocols \cite{pawlowski2011semi,woodhead2015secrecy,ioannou2022receiver}) and DI protocols \cite{Acin2007,Nadlinger2022} respectively. More generally, the CMU can also be seen as a part of more complex scenarios, e.g. involving quantum networks or quantum circuits (see Fig.~\ref{fig:exp}.c). 

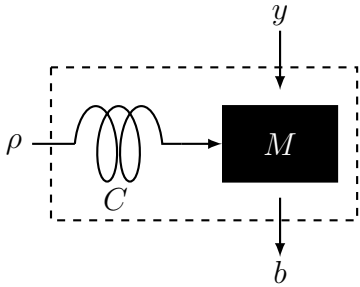
\begin{figure}[t]
    \centering
    \begin{tikzpicture}[auto, thick, node distance=2.5cm, >=latex]
    
    \node [draw, fill=black, minimum width=1.5cm, minimum height=1cm] (box) {\large \textcolor{white}{$M$}};
    \node at (box.center) {};
    
    \node [above of=box, node distance=1cm] (input) {};
    \node [below of=box, node distance=1cm] (output) {};
    
    \node [left of=box, node distance=3.5cm] (rho) {\large $\rho$};
    
    \draw [dashed] ($(box.north west)+(-2.25,0.5)$) rectangle node {} ($(box.south east)+(0.25,-0.5)$);
    
    \node at ($(box.north west)+(-1.4,-1.25)$) {\large $C$};
    
    \draw [->] ($(input)+(0,0.5)$) -- ($(box.north)+(0,0.2)$);
    \draw [->] ($(box.south)+(0,-0.2)$) -- ($(output)-(0,0.5)$);
    \draw [-] (rho) -- (-2.7,0);
    \draw [->] (-1.3,0) -- (box.west);
    
    \draw [decorate, decoration={coil,     aspect=0.4, amplitude=5 mm,  segment length=3mm}, black, thick] (-2.7,0) -- (-1.3,0);
    
    \node [above of=box, node distance=1.7cm] {\large $y$};
    \node [below of=box, node distance=1.7cm] {\large $b$};
    
    \end{tikzpicture}
    \caption{A channel-measurement unit (CMU) is composed of an untrusted channel $C$ and an untrusted measurement device $M$. The CMU receives as input a quantum state $\rho$ that is transmitted through the channel $C$ and then measured by $M$. The classical input $y$ denotes the choice of measurement (performed e.g. by the honest user), while its output is denoted $b$. }\label{fig:CMU}
\end{figure}

We aim to determine the admissible level of loss in the CMU before it loses its ability to exhibit `quantumness', i.e. when it becomes classically simulable. We provide several natural definitions for this, which are shown to be all equivalent to the condition that the effective POMVs implemented by the CMU, representing the combined effect of the channel and the measurement device, are jointly measurable. An equivalent condition, based on the channel extendibility of the CMU, is also particularly convenient to work with. Using these definitions, we provide several upper bounds on the loss threshold. Notably, these bounds typically depend only on the number of measurements implemented by the CMU, and on the noise level, but not on the details of the channel and/or the measurements. 

More generally, our work motivates the investigation of joint measurability for sets of measurements that are not only noisy, the case most existing works have focused on \cite{Quintino2014,uola_joint_2014,Heinosaari2015,Bavaresco2017,Designolle2019,Designolle2019njp}, but that also feature losses, which has been considered only in a limited number of previous works \cite{Skrzypczyk2015,Ioannou2022,sekatski2023unlimited,sekatski2023compatibility}. 

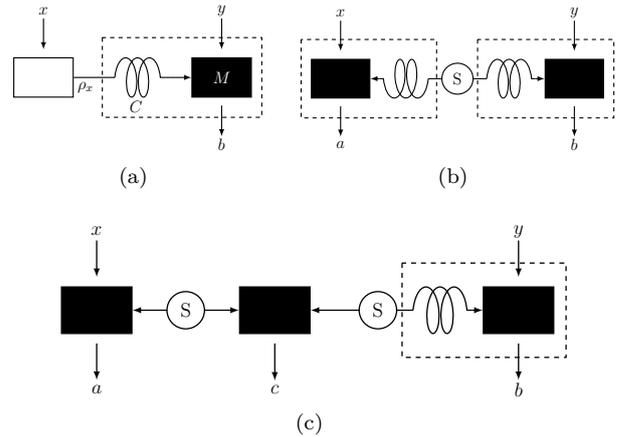
\begin{figure}[t!]
\centering
\subfloat[]{
\resizebox{0.2\textwidth}{!}{
\begin{tikzpicture}[auto, thick, node distance=4.5cm, >=latex]

\node [draw, fill=black, minimum width=1.5cm, minimum height=1cm] (box) {\large \textcolor{white}{$M$}};
\node at (box.center) {};

\node [above of=box, node distance=1cm] (input) {};
\node [below of=box, node distance=1cm] (output) {};

\node [draw, left of=box, fill=white, minimum width=1.5cm, minimum height=1cm] (prep) {};
\node [above of=prep, node distance=1.7cm] {\large $x$};
\node [above of=prep, node distance=1cm] (inputx) {};
\draw [->] ($(inputx)+(0,0.5)$) -- ($(prep.north)+(0,0.2)$);

\draw [dashed] ($(box.north west)+(-2.25,0.5)$) rectangle node {} ($(box.south east)+(0.25,-0.5)$);

\node at ($(box.north west)+(-1.4,-1.25)$) {\large $C$};

\draw [->] ($(input)+(0,0.5)$) -- ($(box.north)+(0,0.2)$);
\draw [->] ($(box.south)+(0,-0.2)$) -- ($(output)-(0,0.5)$);
\draw [-] (prep) -- (-2.7,0);
\draw [->] (-1.3,0) -- (box.west);
\node at ($(prep.east)+(0.3,-0.25)$) {\large $\rho_x$};

\draw [decorate, decoration={coil,     aspect=0.4, amplitude=5 mm,  segment length=3mm}, black, thick] (-2.7,0) -- (-1.3,0);

\node [above of=box, node distance=1.7cm] {\large $y$};
\node [below of=box, node distance=1.7cm] {\large $b$};

\end{tikzpicture}
}}\hfill
\subfloat[]{
\resizebox{0.25\textwidth}{!}{
\begin{tikzpicture}[auto, thick, node distance=3cm, >=latex]

\node (rho) {\large S};
\draw (rho) circle (0.4);

\node [draw, fill=black, minimum width=1.5cm, minimum height=1cm, right of=rho, node distance=3cm] (box) {};
\node at (box.center) {};

\node [above of=box, node distance=1cm] (input) {};
\node [below of=box, node distance=1cm] (output) {};

\draw [dashed] ($(box.north west)+(-1.72,0.5)$) rectangle node {} ($(box.south east)+(0.25,-0.5)$);

\draw [->] ($(input)+(0,0.5)$) -- ($(box.north)+(0,0.2)$);
\draw [->] ($(box.south)+(0,-0.2)$) -- ($(output)-(0,0.5)$);
\draw [-] (0.4,0) -- (0.75,0);
\draw [->] (2,0) -- (box.west);

\draw [decorate, decoration={coil, aspect=0.4, amplitude=5 mm,  segment length=3mm}, black, thick] (0.75,0) -- (2,0);

\node [draw, fill=black, minimum width=1.5cm, minimum height=1cm, left of=rho, node distance=3cm] (whitebox) {};
\node at (whitebox.center) {};

\draw [dashed] ($(whitebox.north west)+(-0.25,0.5)$) rectangle node {} ($(whitebox.south east)+(1.7,-0.5)$);

\node [above of=whitebox, node distance=1cm] (input2) {};
\node [below of=whitebox, node distance=1cm] (output2) {};

\draw [->] ($(input2)+(0,0.5)$) -- ($(whitebox.north)+(0,0.2)$);
\draw [->] ($(whitebox.south)+(0,-0.2)$) -- ($(output2)-(0,0.5)$);
\draw [-] (-0.4,0) -- (-0.75,0);
\draw [->] (-2,0) -- (whitebox.east);

\draw [decorate, decoration={coil, aspect=0.4, amplitude=5 mm,  segment length=3mm}, black, thick] (-0.75,0) -- (-2,0);

\node [above of=box, node distance=1.7cm] {\large $y$};
\node [below of=box, node distance=1.7cm] {\large $b$};
\node [above of=whitebox, node distance=1.7cm] {\large $x$};
\node [below of=whitebox, node distance=1.7cm] {\large $a$};

\end{tikzpicture}
}}\hfill
\subfloat[]{
\resizebox{0.4\textwidth}{!}{
\begin{tikzpicture}[auto, thick, node distance=3cm, >=latex]

\node (rho) {\large S};
\draw (rho) circle (0.4);

\node [draw, fill=black, minimum width=1.5cm, minimum height=1cm, right of=rho, node distance=3cm] (box) {};
\node at (box.center) {};

\node [above of=box, node distance=1cm] (input) {};
\node [below of=box, node distance=1cm] (output) {};

\draw [dashed] ($(box.north west)+(-1.72,0.5)$) rectangle node {} ($(box.south east)+(0.25,-0.5)$);

\draw [->] ($(input)+(0,0.5)$) -- ($(box.north)+(0,0.2)$);
\draw [->] ($(box.south)+(0,-0.2)$) -- ($(output)-(0,0.5)$);
\draw [-] (0.4,0) -- (0.75,0);
\draw [->] (2,0) -- (box.west);

\draw [decorate, decoration={coil, aspect=0.4, amplitude=5 mm,  segment length=3mm}, black, thick] (0.75,0) -- (2,0);

\node [draw, fill=black, minimum width=1.5cm, minimum height=1cm, left of=rho, node distance=6cm] (whitebox) {};
\node at (whitebox.center) {};

\node [right of=whitebox, node distance=1.9cm] (s2) {\large S};
\draw [->] (-4.5,0) -- (whitebox);
\draw (s2) circle (0.4);

\node [draw, fill=black,minimum width=1.5cm, minimum height=1cm, left of=rho, node distance=2.2cm] (node) {};
\draw [->] (-3.7,0) -- (node);
\draw [->] ($(node.south)+(0,-0.2)$) -- ($(node.south)-(0,1.)$);

\node [above of=whitebox, node distance=1cm] (input2) {};
\node [below of=whitebox, node distance=1cm] (output2) {};

\draw [->] ($(input2)+(0,0.5)$) -- ($(whitebox.north)+(0,0.2)$);
\draw [->] ($(whitebox.south)+(0,-0.2)$) -- ($(output2)-(0,0.5)$);
\draw [->] (-0.4,0) -- (node);

\node [above of=box, node distance=1.7cm] {\large $y$};
\node [below of=box, node distance=1.7cm] {\large $b$};
\node [above of=whitebox, node distance=1.7cm] {\large $x$};
\node [below of=whitebox, node distance=1.7cm] {\large $a$};
\node [below of=node, node distance=1.7cm] {\large $c$};

\end{tikzpicture}
}}\hfill
    \caption{The CMU is a basic component of many experiments with untrusted devices, such as (a) a simple prepare-and-measure scenario, (b) A Bell test, (c) a quantum network featuring several sources and measurements}\label{fig:exp}
\end{figure}

We consider two different scenarios for taking losses into account in protocols with untrusted devices. The first consists of attributing a specific measurement outcome (an inconclusive `no-click' outcome) to events where the photons are lost. We refer to this as the \emph{no-click scenario}, corresponding e.g. to experiments involving single-photon detectors. The second consists of general optical scenarios where measurements may always produce a conclusive outcome, as, e.g., in continuous variable setups based on homodyne (or heterodyne) measurements. In this case, we consider a standard model of loss for the optical channel, corresponding to mixing the input system with a thermal state through a beam splitter of limited transmissivity. We refer to this as the \emph{thermal-noise channel scenario}.

In the no-click scenario, we introduce general bounds on critical loss thresholds, which, in their simplest form, recover or improve on certain known bounds \cite{massar_violation_2003,acin2016necessary}. In particular, we prove a bound on $N$-extendibility of channels combining white noise and loss. In the thermal-noise channel scenario, we point out that it is possible to apply directly existing bounds for channel extendibility \cite{lami2019} and joint-measurability \cite{rahimi-keshari_verification_2021}, but we also provide in addition explicit strategies for achieving these bounds. 

Finally, we examine in more detail the applications of our results to QKD. In particular, analogously to \cite{acin2016necessary}, we point out that in a QKD protocol where one of the parties uses a CMU, no key rate can be extracted already when the subset of measurements used to generate the key is jointly measurable. This leads us to introduce a weaker notion of joint measurability which we call `partial joint-measurability'.

The paper is structured as follows. In Section~\ref{sec:JM}, we introduce the CMU and 
make it clear that the concept of joint-measurability is appropriate for defining at what point a CMU loses its utility in a quantum protocol with untrusted devices and we relate it to other properties of the CMU, such as channel extendibility. In Section~\ref{sec:DV}, we derive upper-bounds in the no-click scenario and in Section~\ref{sec:saa} in the thermal-noise channel scenario. In Section~\ref{sec:qkd}, we discuss the applications of our bounds to QKD, define the notion of `partial joint-measurability', and provide as an illustration upper-bounds on the key rate for various classes of DI QKD protocols. Finally, we conclude in Section~\ref{sec:discussion}.

\section{\label{sec:JM} Joint-measurability of a CMU}
A CMU consists of a quantum channel $C$ followed by a measurement device $M$, as in Fig.~1. The channel is described by CPTP map $\Phi$, while the measurement is given by a set of POVMs $M_y =\{M_{b|y}\}_b$, where $M_{b|y} \geq 0$ and $\sum_b M_{b|y} = \openone$ for all $y,b$. Note that we use the notation $M_y$ to denote the POVMs and $M_{b|y}$ their elements.

From a black-box perspective, the CMU is characterized by the probabilities $P(b|y,\rho)$ of obtaining a measurement outcome $b$, given the measurement input $y$ and the incoming quantum state $\rho$. Formally we have that
\begin{equation}
    P(b|y,\rho) = \Trace{\left[\Phi(\rho)\, M_{b|y}\right]}\,.    
\end{equation}
The CMU can also be described through a set of effective POVMs with elements $M^\Phi_{b|y} = \Phi^*\left(M_{b|y}\right)$, where $\Phi^*$ is the dual channel of $\Phi$. These POVMs describe the combined effect of the channel and of the measurements. The probabilities $P(b|y,\rho)$ can then also be written as
\begin{equation}
    P(b|y,\rho) = \Trace{\left[\rho\,  M^\Phi_{b|y}\right]}\,.    
\end{equation}

We would like to determine the admissible level of loss in the CMU, regardless of its usage in a larger experiment, and in a way that is independent of the state preparation, before it loses its ability to exhibit `quantumness'. Since most experiments and protocols with untrusted devices require the distribution of entanglement and/or incompatible measurements, we could define this threshold using one of the two following natural conditions:
\begin{enumerate}
    \item[(1)] The CMU can be replaced by an equivalent CMU for which the channel $\Phi$ is entanglement-breaking.
    \item[(2)] The CMU can be replaced by an equivalent CMU for which the $N$ measurements $M_{y}$ are compatible, i.e., jointly measurable.
   \end{enumerate}
By `equivalent' CMU, we mean one that may differ in implementation from the original CMU, but yields the same output probabilities $P(b|y,\rho)$ for all $\rho$, since from an operational and blackbox perspective, only these output probabilities are relevant.

As an example of the above definition, consider a standard bipartite nonlocality test.  If either Alice's or Bob's CMU satisfies condition (1) or (2), it becomes clear that the entire experiment admits a local hidden-variable theory and no violation of a Bell inequality is possible. Similar conclusions can be drawn for steering experiments. Moreover, even protocols that do not directly rely on entanglement may require condition (1) to hold. For example, in a prepare-and-measure quantum key distribution (QKD) scheme where the channel connecting Alice to Bob is entanglement-breaking, no key rate can be extracted \cite{curty_entanglement_2004}.

More generally, the purpose of a CMU is to distribute quantum information from the entry of the channel to a remote measurement device where one of several measurements should be performed. Then an alternative, and seemingly more stringent, way to define the threshold at which a CMU becomes ineffective is the following:
\begin{enumerate}
    \item[(3)] The CMU can be replaced by an equivalent CMU for which the channel $\Phi$ is a quantum-classical channel and the measurements $M_{y}$ represent classical measurements on the output of this channel.
\end{enumerate}
In a quantum experiment involving a CMU that satisfies condition (3), the quantum part of the experiment can be truncated just before the CMU, since all information processing that happens in the CMU is purely classical, see Fig.~\ref{fig:classical}. This criterion is used, for instance, in \cite{lobo2023certifying} to determine when a routed Bell experiment can establish only short-range quantum correlations. Furthermore, in any cryptographic protocol where the CMU channel is a public channel and satisfies condition (3), an eavesdropper could have a perfect copy of the information sent in that channel since it is purely classical.

\begin{figure}[t]
\centering
    \begin{tikzpicture}[auto, thick, node distance=2.5cm, >=latex]
    
    \node [draw, minimum width=1.5cm, minimum height=0.8cm] (laptop) {};
    \node at (laptop.center) {};
    \draw ([xshift=-0.62cm]laptop.south) rectangle ([xshift=0.65cm,yshift=-0.5cm]laptop.south);
    
    \draw ([xshift=-0.7cm,yshift=0.35cm]laptop) rectangle ([xshift=0.7cm,yshift=-0.35cm]laptop);
    
    \foreach \x in {0,...,5}
        \foreach \y in {0,...,2}
            \draw ([xshift=-0.55cm + \x * 0.2cm,yshift=-0.15cm - \y * 0.15cm]laptop.south) rectangle ++(0.15, 0.1);
    
    \node [above of=laptop, node distance=1cm] (input) {};
    \node [below of=laptop, node distance=1cm] (output) {};
    
    \node [left of=laptop, node distance=6cm] (rho) {\large $\rho$};
    
    \draw [dashed] ($(laptop.north west)+(-4.5,0.5)$) rectangle node {} ($(laptop.south east)+(0.25,-0.8)$);
    
    \node [draw, fill=black, minimum width=0.9cm, minimum height=0.9cm, right of=rho, node distance=1.5cm] (meas) {\textcolor{white}{E}};
    
    \draw [->] ($(input)+(0,0.5)$) -- ($(laptop.north)+(0,0.2)$);
    \draw [->] ($(laptop.south)+(0,-0.5)$) -- ($(output)-(0,0.8)$);
    \draw [->] (rho) -- (meas.west);
    \draw [->] (-1.8,0) -- (laptop.west);
    
    \node[antenna,xscale=.8] (antenna) at (-2.5,-0.25){};
    \fill[black] ($(antenna.north)$) circle[radius=.7mm];
    \foreach \r in {2mm,3.5mm,5mm}{
        \draw ($(antenna.north)+(-30:\r)$) arc[start angle=-30,end angle=50,radius=\r];
    }
    \foreach \r in {2mm,3.5mm,5mm}{
        \draw ($(antenna.north)+(130:\r)$) arc[start angle=130,end angle=210,radius=\r];
    }
    
    \node at ($(meas.south east)+(0.45,-0.45)$) (lambda) {$\lambda$};
    \draw[->] (meas.south) to [out=-90,in=-180] (lambda.west);
    \draw[->] (lambda.east) to [out=0,in=-90] ($(antenna.south)-(0,0.1)$);
    
    \node [above of=laptop, node distance=1.7cm] {\large $y$};
    \node [below of=laptop, node distance=2cm] {\large $b$};
    
    \end{tikzpicture}
    \caption{A CMU that satisfies conditions (3) and (4) can be implemented through a (single) quantum measurement $E_{\lambda}$ at the entrance of the channel, whose classical output $\lambda$ is then transmitted and  processed to generate the output $b$ depending on $y$.}\label{fig:classical}
\end{figure}
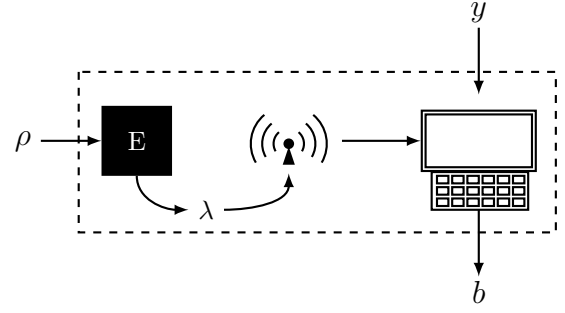

Although conditions (1), (2), (3) may appear distinct, they are actually all equivalent  in an untrusted scenario,  and equivalent to the following condition:
\begin{enumerate}
    \item[(4)] The set of effective POVMs $M^\Phi_{b|y}$ of the CMU are jointly measurable.
\end{enumerate}

The set of effective POVMs $M^\Phi_{b|y}$ is said to be jointly measurable \cite{uola2023} if there exists a \emph{parent POVM} $E$ with elements $E_\lambda$, and probability distributions $p(b|y,\lambda)$ such that
\begin{equation}\label{eq:jm}
    M^\Phi_{b|y} = \int_\lambda d\lambda\, p(b|y,\lambda) E_\lambda\,.
\end{equation}
Operationally, and as illustrated in Fig.~\ref{fig:classical}, this means that the CMU can in principle be implemented via the following strategy. First, at the entrance of the channel, upon receiving the quantum state $\rho$, one measures the parent POVM $E$; note that this is a single (fixed) quantum measurement, hence independent of the input $y$. The outcome $\lambda$ of the parent POVM is then transmitted through a classical channel to the remote measurement device, where a simple classical device generates the output $b$ using the probability distribution $p(b|y,\lambda)$, which depends on $y$ and on the classical outcome $\lambda$.

From the above operational interpretation of (4), it is immediate that conditions (3) and (4) are equivalent. It is also clear that conditions (2) and (4) are equivalent, since if the measurements $M_{y}$ are jointly measurable, the same holds for the effective measurements $M^\Phi_y$, and vice versa. Finally, conditions (1) and (4) are also equivalent, as shown in \cite{pusey_verifying_2015}. Indeed,  $\Phi$ is entanglement-breaking if $\Phi(\rho) = \sum_\lambda\sigma_\lambda \Trace{\left[E_\lambda \rho\right]} $. But then $P(b|y,\rho) = \sum_\lambda \Trace{\left[E_\lambda \rho\right]} \Trace{\left[\sigma_\lambda M_{b|y}\right]} = \Trace{\left[\rho M^\Phi_{b|y}\right]}$, for $M^\Phi_{b|y} = \sum_\lambda \Trace{\left[\sigma_\lambda M_{b|y}\right]} E_\lambda$ which is of the form \eqref{eq:jm} if we define $P(b|y,\lambda) = \Trace{\left[\sigma_\lambda M_{b|y}\right]}$. 
This argument can be extended to infinite-dimensional quantum systems; in this case, an entanglement breaking channel is of the form $\Phi(\rho) = \int \dd \lambda\, \sigma_\lambda \Tr \rho E_\lambda $ where the POVM is in general not atomic~\cite{holevo2005separability}.

Note that condition (4) is not only sufficient for guaranteeing that a CMU is ``useless'', but also a necessary one; indeed, any set of measurements that is incompatible (in the sense of not being jointly measurable) can be used to demonstrate the effect of quantum steering \cite{Quintino2014,uola_joint_2014}. 

In the following, we will thus take condition (4) as our definition of when a CMU does not exhibit quantumness. We say that a CMU is \emph{jointly-measurable} (\JM{}) whenever condition (4) holds.

Finally, note that there are other possible conditions than (1) on the channel of a CMU that are equivalent to (4). For instance, a channel is said to be \emph{incompatibility-breaking} if the set of measurements $\{\Phi^*(M_{y})\}$ are compatible for any $M_y$ \cite{heinosaari2015incompatibility}. It is said to be 
$N$-\textit{extendable} if there exists another 
channel $\Phi_{1\to N}$, called the parent channel, which outputs $N$ quantum system $S_1\dots,S_N$ such that
\begin{equation}\label{eq: marginal channels}
    \Tr_{S_\mu|\mu\neq k}\Phi_{1\to N}(\rho) = \Phi(\rho)
\end{equation}
for any output system $S_k$ and any input state $\rho$. Let us recall that $N$ is the number of measurements performed by our CMU. These notions give rise to the following additional conditions.
\begin{enumerate}
    \item[(5)] The CMU can be replaced by an equivalent CMU for which the channel $\Phi$ is incompatibility-breaking.
    \item[(6)] The CMU can be replaced by an equivalent CMU for which the channel $\Phi$ is $N$-extendable.
\end{enumerate}

The entanglement-breaking, incompatibility-breaking, and $N$-extendibility of a channel are in principle distinct properties. While every entanglement-breaking channel is $N$-extendable, and any $N$-extendable channel is incompatibility breaking, there exist incompatible channels that are not $N$-extendable, and $N$-extendable channels that are not entanglement-breaking. However, in the context of an untrusted CMU, the conditions (1), (5), and (6) are all equivalent to (4). This is because they are not statements about the properties of the channel in the actual realization of the CMU, but they are statements about an operationally equivalent CMU that reproduces the \emph{combined} effect of the \emph{channels} and of the \emph{measurements} of the original CMU.

Specifically, since a classical channel is both incompatibility-breaking and $N$-extendable , we clearly have (3)$\Rightarrow$(5) and (3)$\Rightarrow$(6), and thus (4)$\Rightarrow$(5) and (4)$\Rightarrow$(6). It is also clear that (5)$\Rightarrow$(4). On the other hand, if (6) holds, i.e., the channel $\Phi$ is $N
$-extendable, the effective POVMs $\Phi^*(M_{b|y})$ are jointly-measurable because they can be obtained by applying the channel $\Phi_{1\to N}$ on the input state, and then performing each of the $N$ different measurement $M_{y}$ for $y=1,\dots N$ on each of the $N$ output systems, as illustrated in Fig.~\ref{fig:cmu5}. Thus (6)$\Rightarrow$(4). We will make extensive use of this equivalence between the \JM{} property of the CMU and $N$-extendibility in the remainder of the paper.
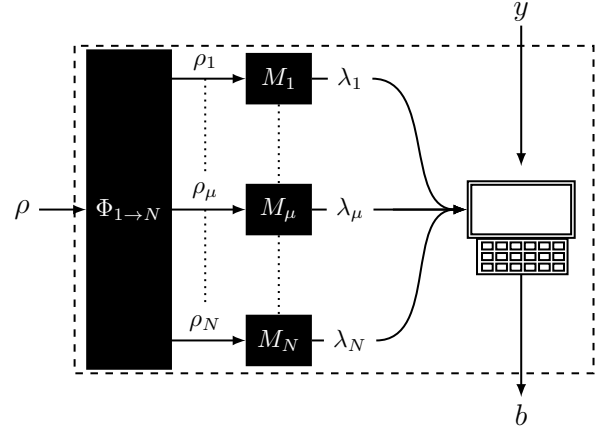
\begin{figure}[t]
    \centering
    \resizebox{.95\linewidth}{!}{\begin{tikzpicture}[auto, thick, node distance=2.5cm, >=latex]
    
    \node [draw, minimum width=1.5cm, minimum height=0.8cm] (laptop) {};
    \node at (laptop.center) {};
    \draw ([xshift=-0.62cm]laptop.south) rectangle ([xshift=0.65cm,yshift=-0.5cm]laptop.south);
    
    \draw ([xshift=-0.7cm,yshift=0.35cm]laptop) rectangle ([xshift=0.7cm,yshift=-0.35cm]laptop);
    
    \foreach \x in {0,...,5}
        \foreach \y in {0,...,2}
            \draw ([xshift=-0.55cm + \x * 0.2cm,yshift=-0.15cm - \y * 0.15cm]laptop.south) rectangle ++(0.15, 0.1);
    
    \node [above of=laptop, node distance=1cm] (input) {};
    \node [below of=laptop, node distance=1cm] (output) {};
    
    \node [left of=laptop, node distance=7cm] (rho) {\large $\rho$};
    
    \draw [dashed] ($(laptop.north west)+(-5.5,1.9)$) rectangle node {} ($(laptop.south east)+(0.25,-1.9)$);
    
    \node [draw, fill=black, minimum width=0.9cm, minimum height=4.5cm, right of=rho, node distance=1.5cm] (meas) {\textcolor{white}{$\Phi_{1\rightarrow N}$}};

    \node[draw, fill=black, minimum width=0.9cm, minimum height=0.7cm] (m1) at (-3.4,1.85) {\textcolor{white}{$M_1$}};
    \node[draw, fill=black, minimum width=0.9cm, minimum height=0.7cm] (m2) at (-3.4,0) {\textcolor{white}{$M_\mu$}};
    \node[draw, fill=black, minimum width=0.9cm, minimum height=0.7cm] (m3) at (-3.4,-1.85) {\textcolor{white}{$M_N$}};
    
    \node (l1) at (-2.4,1.85) {$\lambda_1$};
    \node (l2) at (-2.4,0) {$\lambda_\mu$};
    \node (l3) at (-2.4,-1.85) {$\lambda_N$};
    
    \draw [->] ($(input)+(0,1.6)$) -- ($(laptop.north)+(0,0.2)$);
    \draw [->] ($(laptop.south)+(0,-0.5)$) -- ($(output)-(0,1.65)$);
    \draw [->] (rho) -- (meas.west);
    \draw [->] (-1.8,0) -- (laptop.west);

    \draw[->] (-5,1.85) -- (m1.west) node[midway,above] (r1) {$\rho_1$};
    \draw[->] (-5,0) -- (m2.west) node[midway,above] (r2) {$\rho_\mu$};
    \draw[->] (-5,-1.85) -- (m3.west) node[midway,above] (r3) {$\rho_N$};
    
    \draw[-] (m1.east) -- (l1.west);
    \draw[-] (m2.east) -- (l2.west);
    \draw[-] (m3.east) -- (l3.west);
    
    \draw[-,dotted] (m1.south) -- (m2.north);
    \draw[-,dotted] (m2.south) -- (m3.north);
    \draw[-,dotted] (r1.south) -- (r2.north);
    \draw[-,dotted] (r2.south) -- (r3.north);
    
    \draw[->] (l1.east) to [out=0,in=180] (laptop.west);
    \draw[->] (l2.east) -- (laptop.west);
    \draw[->] (l3.east) to [out=0,in=180] (laptop.west);
    
    \node [above of=laptop, node distance=2.8cm] {\large $y$};
    \node [below of=laptop, node distance=2.9cm] {\large $b$};
    
    \end{tikzpicture}}
    \caption{A CMU satisfying condition (6) can be implemented by splitting the incoming system in $N$ systems, measuring each of the output systems with one of the $N$ different measurements, and sending their classical outcomes to Bob.}
    \label{fig:cmu5}
\end{figure}

\section{No-click CMUs}\label{sec:DV}

We first analyze a large class of CMUs taking as input quantum states $\rho$ in a finite $d$-dimensional Hilbert space $\mathcal{H}$ and for which the $N$ measurements $y=1,\ldots,N$ yield a finite, discrete set of $L+1$ outcomes $b=1,\ldots,L,\varnothing$, where the last outcome $\varnothing$ is a `no-click' outcome happening when the measurement device fails to register an actual outcome, e.g. because the photon is lost. In addition to the loss, the measurements can also be affected by another type of noise.

In general, such a CMU is described by effective POVMs of the form 
\begin{equation}\label{eq:effective}
     M_{b|y}^{\eta,\Phi}\!  = 
    \begin{cases}
    \eta\, \Phi^*\!\left(M_{b|y}\right) & \text{if } b \neq \varnothing \\
    (1-\eta) \mathbbm{1} & \text{if } b=\varnothing\,,
    \end{cases}
\end{equation}
where $M_{b|y}$ describe the ideal POVMs, $\eta$ is the \emph{detecting efficiency}, and the CPTP map $\Phi:\mathcal{L}(\mathcal{H})\rightarrow \mathcal{L}(\mathcal{H})$ models the additional noise channel ($\Phi^*$ is its adjoint). This CMU can be understood as a concatenation of the channel $\Phi$ and the \emph{loss channel}
\begin{equation}\label{eq: loss channel}
\Lambda_\eta( \rho ) = \eta \rho + (1-\eta)\ketbra{\varnothing}
\end{equation}
acting on the state $\rho$ before the measurements $\tilde M_y$ with elements 
\begin{equation}\label{eq:ext-povms}
     \tilde M_{b|y}\!  =   
    \begin{cases} M_{b|y}&\text{if } b \neq \varnothing \\
    \ketbra{\varnothing} & \text{if } b=\varnothing
    \end{cases} 
\end{equation}
are performed. Here, $\ket{\varnothing}$ is an additional state orthogonal to $\mathcal{H}$ describing the loss of a photon, and (\ref{eq:ext-povms}) extends the ideal POVMs so as to produce the no-click outcome whenever the state $\ket{\varnothing}$ is received. This allows one to write 
$M_{b|y}^{\eta,\Phi}\! = (\Lambda_\eta \circ \Phi)^*(\tilde M_{b|y})$. 
Note that whether the no-click outcome is deterministically assigned to another outcome or treated as an additional outcome does not impact the results, as this is a classical post-processing step. We will discuss again this scenario in Section~\ref{sec:binning} in the context of QKD.

For most of our applications, we will consider the noise channel $\Phi$ to be the \emph{white-noise} channel 
\begin{equation}\label{eq:wn-channel}
    W_v(\rho) = v\, \rho + (1-v)\frac{\mathbbm{1}}{d}\,,
\end{equation}
with \emph{visibility} $v$. In this case, the effective POVMs of the CMU read
\begin{equation}\label{eq:wn-povms}
    M_{b|y}^{\eta,v} = 
    \begin{cases}
        \eta\,v\,M_{b|y}+\eta\,(1-v) q_{b|y}\,\mathbbm{1} & \text{if } b\neq \varnothing\\
        (1-\eta)\mathbbm{1} & \text{if } b=\varnothing\,,
    \end{cases}
\end{equation}
where $q_{b|y} = \Tr\left[M_{b|y}\right]/d$.\\

Given the ideal POVMs $M_{b|y}$, the noise channel $\Phi$ and the efficiency $\eta$, the question of whether the effective POVMs (\ref{eq:effective}) are \JM{} can be formulated as a semidefinite program (SDP) \cite{Wolf2009,uola2023}.  Furthermore, the threshold efficiency $\eta_*$ below which they are \JM{} can readily be determined by solving a single SDP since the above effective POVMs depend linearly on $\eta$.

In the following, we are interested in providing bounds on the critical values of $\eta$ that hold for a large class of measurements $M_{y}$ and noise models $\Phi$. 

\subsection{Bounds on $\eta$ for arbitrary sets of measurements implied by channel extendibility}
In this section, we first collect some results on the $N$-extendibility of some families of channels. As discussed in the introduction, we then use them to imply that any CMU featuring such a channel and composed of $N$ POVMs is \JM{}. We start by introducing two general parametric families of channels. 

A  channel $R_q$ with $q\in[0,1]$ is called  \emph{replacement} if it is of the form 
\begin{equation}\label{eq: remp}
        R_{q}(\rho) = q\, \rho +(1-q)\sigma,
    \end{equation}
where $\sigma$ is any fixed state (possibly orthogonal to $\mathcal{H}$).
The loss channel $\Lambda_\eta$ in Eq.~\eqref{eq: loss channel} and the white-noise channel $W_v$ in Eq.~\eqref{eq:wn-channel} are instances of replacement channels. 

A channel $\Gamma_v$ with $v\in [0,1]$  is called \emph{convex noise} if it is of the form
\begin{equation}\label{eq: noise-convex}
    \Gamma_v(\rho) = v \rho + (1-v) \Gamma_0(\rho),
\end{equation}
for some $\Gamma_0$. An equivalent definition is to impose the properties $\Gamma_1 = \text{id}$ and $ p \, \Gamma_{v_1} + (1-p) \, \Gamma_{v_2} = \Gamma_{\bar v}$, where $\bar v = p \, v_1 +(1-p)v_2$. It is easy to see that all the replacement channels are convex noise channels. However, there are other prominent examples of convex noise channels, e.g. the dephasing channel
\begin{equation}\label{eq: dephasing}
    \Gamma_v(\rho) = v \rho + (1-v) \text{diag}\,(\rho)\,,
\end{equation}
where $\text{diag}(\rho) = \sum_i \rho_{ii} \ketbra{i}$ is the diagonal part of $\rho$ in the computational basis.\\

With these definitions, we now formulate a few results on $N$-extendibility. The first one holds for all replacement channels.

\begin{lem}[\cite{heinosaari2015incompatibility}]\label{Lemma:remp}
    A replacement channel $R_{q}(\rho)$ as in Eq.~\eqref{eq: remp} 
    is $N$-extendable if $q\leq \frac{1}{N}$.
\end{lem}
\begin{proof}
For $q\leq \frac{1}{N}$ the channel
\begin{align}\label{eq:global}
    R_{1\to N}(\rho) &= q \,\rho \otimes \sigma^{\otimes(N-1)}+ q\, \sigma\otimes \rho\otimes \sigma^{\otimes(N-2)}\nonumber\\
    &+\ldots+q\, \sigma^{\otimes(N-1)}\otimes \rho
     + (1-qN)\sigma^{\otimes N}\,.
\end{align}
is well defined and is the parent of $R_q$.
\end{proof}

This directly implies that the loss channel $\Lambda_\eta$ in Eq.~\eqref{eq: loss channel} is $N$-extendable for $\eta\leq \frac{1}{N}$. A similar implication holds for the white-noise channel $W_v$ in Eq.~\eqref{eq:wn-channel}, however for this channel the following result gives a tighter bound. 

\begin{lem}[\cite{werner1998optimal}]\label{Lemma:wn}
    The white-noise channel $W_v$ in Eq.~\eqref{eq:wn-channel} is $N$-extendable if
    \begin{equation}\label{eq:vis}
    v \leq \frac{N+d}{N(d+1)}.
\end{equation}
\end{lem}
\begin{proof}
This follows from results on universal quantum cloning. An optimal one-to-$N$ universal cloning machine is a channel $\Phi_{1\to N}^\text{UCM}$ from one to $N$ qudits whose output resembles $N$ copies of the input state as closely as possible for all possible input states~\cite{werner1998optimal,keyl1999}. Furthermore, for each of the output systems, the marginal channels of Eq.~\eqref{eq: marginal channels} induced by the optimal universal cloning machine are precisely given by a white-noise channel $W_{v}$ (implied by the symmetry of the task) with visibility saturating the inequality~\eqref{eq:vis}~\cite{werner1998optimal}. More noise can always be added if $v$ is lower (see Lemma ~\ref{Lemma:conc general} below). 
\end{proof}

Now, following our original motivation we proceed to study the $N$-extendibility of concatenated channels. The next result is a very simple but useful observation.

\begin{lem}\label{Lemma:conc general}
    Consider an $N$-extendable channel $\Phi$ and another arbitrary channel $\Gamma$. The concatenated channels 
    \begin{equation} \Gamma\circ \Phi \quad \text{and} \quad \Phi\circ \Gamma
    \end{equation}
    are $N$-extendable.
\end{lem}
\begin{proof}
The parent channels are given by $ \Gamma^{\otimes N}\circ \Phi_{1\to N}$ and $\Phi_{1\to N}\circ \Gamma$  respectively. 
\end{proof}

Now consider a concatenation of channels with known properties.

\begin{lem}\label{Lemma: conc structure}
    Consider a convex noise  channel $\Gamma_v$ (Eq.~\ref{eq: noise-convex}) with $v\geq v_*$ and a replacement channel $R_q$ (Eq.~\ref{eq: remp}) with $q\geq q_*$, where above the threshold values the channels $\Gamma_{v}$ and $R_{q}$ are $N$-extendable. The concatenated channel 
    \begin{equation}
    \Phi_{q, v} = R_q\circ \Gamma_v
    \end{equation}
    is $N$-extendable for  
\begin{equation}\label{eq:gb}
    q\leq \frac{1-v_*}{(1-v)+(v-v_*)/q_*}\,.    
\end{equation}
\end{lem}
\begin{proof}
First, let us write the concatenated channel explicitly as
\begin{equation}\label{eq: 17}\begin{split}
    \Phi_{q, v}(\rho)&= R_q \big( v \rho + (1-v) \Gamma_0(\rho)\big)\\
    & = q \,v \,\rho + q (1-v) \Gamma_0(\rho) + (1-q)\sigma.
\end{split}
\end{equation}
    By assumption, there exist parent channels $\Gamma_{1\to N}$ and $R_{1\to N}$ that reproduce the marginal channels $\Gamma_{v_*}$ and $R_{q_*}$ when tracing out $N-1$ systems.
    Let us now consider a mixture thereof $p \, \Gamma_{1\to N}+ (1-p) \, R_{1\to N} $ for some parameter $p\in[0,1]$. It defines a  parent channel which establishes the $N$-extendibility of $\mathcal{E}_p(\rho) = \tr_{N-1} (p \, \Gamma_{1\to N} + (1-p) R_{1\to N} )(\rho)$ given by 
    \begin{equation}\begin{split}
        \mathcal{E}_p(\rho) & = p\, \Gamma_{v_*}(\rho) + (1-p) R_{q_*}(\rho) \nonumber\\
         & = (p\, v_*+(1-p)q_*) \rho + p (1-v_*) \Gamma_0(\rho) \\
         &+ (1-p) (1-q_*) \sigma.      
    \end{split}
    \end{equation}
    Comparing with Eq.~\eqref{eq: 17} we conclude that $\Phi_{q,v}$ is identical to $\mathcal{E}_p$  if 
    \begin{equation}
        \begin{aligned}
        1-q &= (1-p)(1-q_*)\\
        q(1-v) & =p (1-v_*)\,.
        \end{aligned}
    \end{equation}
    Solving the first equation for $p$ gives $1-p= \frac{1-q}{1-q_*}$. Plugging this in the second equation implies 
    \begin{equation}
          q = \frac{1-v_*}{(1-v)+(v-v_*)/q_*}\,.
    \end{equation}
    The channel $\Phi_{q,v}$ is thus $N$-extendable when this equality is satisfied. To see that any channel with a lower transmission $q'\leq q$ is also $N$-extendable note that it can be decomposed as $\Phi_{q',v} = R_{q'/q} \circ \Phi_{q,v}$ and apply Lemma~\ref{Lemma:conc general}.
\end{proof}

Let us now come back to our channels of interest and apply the above Lemmas to obtain sufficient conditions for the CMUs with effective POVMs \eqref{eq:effective} and \eqref{eq:wn-povms} to become \JM.
First, consider the concatenation of any channel $\Phi$ with the loss channel $\Lambda_\eta$ and apply Lemmas~\ref{Lemma:remp} and \ref{Lemma:conc general} to obtain the following bound.

\begin{ub}[\cite{acin2016necessary}]  \label{ub:1}
    The concatenation of the loss channel with any quantum channel $\Phi$ \begin{equation}
    \Lambda_\eta\circ \Phi(\rho)= \eta\, \Phi(\rho)+ (1-\eta)\ketbra{\varnothing}
    \end{equation}
    is $N$-extendable, and thus the corresponding CMU with effective POVMs (\ref{eq:effective}) is \JM{}, if 
    \begin{equation}\label{eq:eta-no-noise}
        \eta\leq \frac{1}{N}\,.
    \end{equation}
\end{ub}
This is essentially a reformulation from the channel perspective of a result already obtained in \cite{massar_violation_2003} in the context of Bell experiments, in \cite{acin2016necessary} for general one-sided-device-independent protocols, and in \cite{heinosaari2015incompatibility} where channels of the form $R_\eta\circ\Gamma(\rho) = \eta\,\Gamma(\rho)+(1-\eta) \sigma$ were shown to be $N$-incompatibility breaking for $\eta\leq\frac{1}{N}$. 

Next, we also assume that the noise channel $\Phi=\Phi_v$ is convex and apply Lemmas~\ref{Lemma:remp} and \ref{Lemma: conc structure} to get.

\begin{ub}\label{ub:2}
The channel
\begin{equation}
\Lambda_\eta \circ \Phi_v (\rho)= \eta \, \Phi_v(\rho) + (1-\eta)\ketbra{\varnothing},
\end{equation}
corresponding to the concatenation of the loss channel and any convex noise channel $\Phi_v$ with $v\geq v_*$ where $\Phi_{v_*}$ is $N$-extendable, is $N$-extendable, and thus the corresponding CMU is \JM{}, if
\begin{equation}\label{eq:gb2}
    \eta\leq \frac{1-v_*}{(1-v)+N(v-v_*)}\,.    
\end{equation}
\end{ub}
Remarkably, this bound is useful even if we only know that the channel $\Phi_v$ becomes $N$-extendable at $v=0$. Direct application of \eqref{eq:gb2} shows that in this case $\Lambda_\eta \circ \Phi_v$ is $N$-extendable for 
\begin{equation}\label{ub:2.1}
    \eta \leq \frac{1}{(1-v)+N v}.
\end{equation}
This bound applies for example when $\Phi_v=\Gamma_v$ is the dephasing channel in Eq.~\eqref{eq: dephasing}.

Similarly, if $\Phi_v$ is a replacement channel itself (or becomes $N$-extendable for $v\leq \frac{1}{N}$ for another reason) the bound~\eqref{eq:gb2} implies that $\Lambda_\eta \circ \Phi_v$ is $N$-extendable for 
\begin{equation}\label{ub:2.2}
    \eta \leq \frac{1}{v N}.
\end{equation}

\begin{figure}[t]
    \centering
    \resizebox{.95\linewidth}{!}{\begin{tikzpicture}
    \begin{axis}[
            grid=major,
            xtick={1/3,0.5,1/sqrt(3),2/3,1/sqrt(2),1},
            xticklabels={$\frac{1}{3}$,$\frac{1}{2}$,$\frac{1}{\sqrt{3}}$,$\frac{2}{3}$,$\frac{1}{\sqrt{2}}$,$1$},
            extra x ticks={0.556}, 
            extra x tick labels={$\frac{5}{8}$},
            extra x tick style={tick label style={xshift=-3pt}},
            ytick={0,1/3,0.5,1},
            yticklabels={$0$,$\frac{1}{3}$,$\frac{1}{2}$,$1$},
            extra y ticks={0.25}, 
            extra y tick labels={$\frac{1}{4}$},
            extra y tick style={tick label style={xshift=-5pt}},
            ylabel = {$\eta$},
            xlabel = {$v$},
            x label style={yshift=-1pt},
            legend style={at={(0.25,0.1)}, anchor=south},
            xmin=0.3,
            xmax=1.03,
            ymin=-0.05,
            ymax=1.05
          ]
    \addplot[color=blue, thick, domain=(2/3):(1)]{1/(3*x-1)};
    \addlegendentry{$N=2$};
    \addplot[color=orange, thick, domain=(5/9):(1)]{2/3/(3*x-1)};
    \addlegendentry{$N=3$};
    \addplot[color=purple, thick, domain=(1/2):(1)]{1/2/(3*x-1)};
    \addlegendentry{$N=4$};
    \addplot[thick, green!70!black]{3};
    \draw[thick, green!70!black] (axis cs:1/3,0) -- (axis cs:1/3,1);
    \addlegendentry{$N\rightarrow\infty$};
    \addplot[blue, dashed, thick] table {pavelN2.txt};
    \addplot[color=orange, dashed, thick] table {pavelN3.txt};
    \addplot[color=green!70!black, dashed, thick] table {pavelNinf.txt};
    \end{axis}
    \end{tikzpicture}}
    \caption{The full curves depict the lines below and on the left of which the channel $\Phi_{\eta,v}$ of the no-click white-noise qubit CMU is $N$-extendable according to Upper bound~\ref{ub:wn} for $N=2,3,4$ and according to Lemma~\ref{Lemma:wn} for $N=\infty$. The dashed curves correspond to the additional assumption that the ideal measurements are binary and are determined according to Upper bound~\ref{ub:binary} (see section~\ref{sec: no-click qubit}) for $N=2,3,4$ and the bound (\ref{eq:allmeas}) (see section~\ref{sec: no-click qubit}) for $N=\infty$. Note that for the case $N=4$ the two lines coincide.}
    \label{fig:new}
\end{figure}
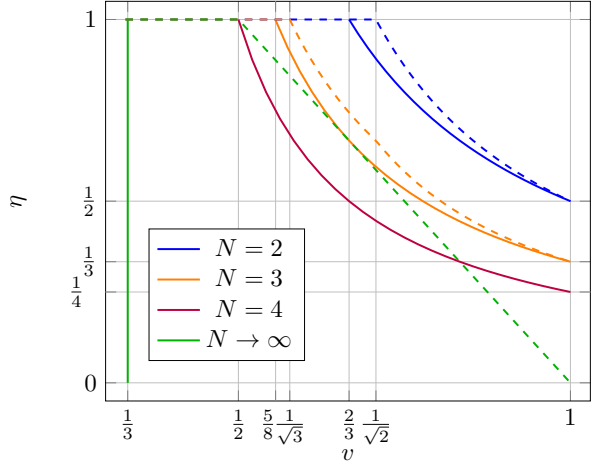

Finally, we give special treatment to the case where $\Phi_v=W_v$ is the white-noise channel, as it will be important for applications. Combining the Upper bound \ref{ub:2} with Lemma \ref{Lemma:wn} (giving $v_*=\frac{N+d}{N(d+1)}$) we get our last bound in this section.
 
\begin{ub}\label{ub:wn}
    The concatenation of loss and white-noise
    \begin{equation}
\Lambda_\eta \circ W_v (\rho)= \eta\,v \,\rho + \eta(1-v)\frac{\mathbbm{1}}{d}  + (1-\eta)\ketbra{\varnothing},
\end{equation}
   is $N$-extendable, and thus the corresponding CMU with effective POVMs \eqref{eq:wn-povms} is \JM{}, if
    \begin{equation}\label{eq:lb-noclick}
         \eta\leq\frac{d}{N\Big(v(d+1)-1\Big)}\,.
    \end{equation} 
\end{ub}

Note that we can make this bound dimension-independent by taking the worst-case value of the right-hand side of (\ref{eq:lb-noclick}),
$\eta\leq \min_{d\geq 1}d/ \big(N(v(d+1)-1)\big) = \frac{1}{Nv}\,$, which coincides with the more general bound (\ref{ub:2.2}).

The fundamental trade-off between the visibility $v$ and the detection efficiency $\eta$ provided by the bound (\ref{eq:lb-noclick}) is illustrated in Fig.~\ref{fig:new} in the case of qubits.

\subsection{Bounds on $\eta$ for binary qubit measurements}
\label{sec: no-click qubit}

The bounds obtained so far on $\eta$ depend on generic parameters such as the number of measurements $N$, the level of noise $v$, or the dimension $d$ of the underlying Hilbert space in the case of white noise. And as they follow from the $N$-extendibility of the CMU channel, they hold for any possible measurements $M_{y}$ implemented by the CMU. We now take additionally into account information about the specific measurements $M_{y}$ that define the CMU, focusing on the particular case of qubits.\bigskip

First note the following variation of Upper bound~\ref{ub:2}.
\renewcommand{\theub}{2'}
\begin{ub}\label{ub:2'}
    Consider a no-click CMU defined in terms of a convex noise channel $\Phi_v$ with $v\geq v_*$ and ideal measurements $M_{y}$ such that the effective POVMs $\Phi_{v_*}^*(M_{y})$ are jointly measurable. Then the effective POVMs  $\Phi_{\eta,v}^*(M_{y})$ defining the no-click CMU in Eq.~\eqref{eq:effective} are \JM{} for 
\begin{equation}\label{eq:gbJM}
    \eta\leq \frac{1-v_*}{(1-v)+N(v-v_*)}\,.    
\end{equation}
\end{ub}
\renewcommand{\theub}{\arabic{ub}}
\setcounter{ub}{5}

This can be proven identically to the  Lemma \ref{Lemma: conc structure} by comparing the POVMs
\begin{align}
    \Phi^*_{\eta,v}(M_{b|y}) &= \eta v M_{b|y} + \eta (1-v) \Phi^*_{1,0}(M_{b|y}) \nonumber\\
    &+ (1-\eta)  \Phi^*_{0,1}(M_{b|y})
\end{align}
with the family of POVMs
\begin{equation}
    p \, \Phi^*_{1,v_*}(M_{b|y})  + (1-p)  \Phi^*_{\frac{1}{N},1}(M_{b|y}),
\end{equation}
which are \JM{} by construction.\\

From now on we set $d=2$ for the rest of the section. A binary qubit measurements is composed of two POVM elements $\{E_+,E_-\}$  of the form
\begin{equation}
E_\pm = \, \frac{1}{2}\big((1\pm \gamma)\mathbbm{1} \pm \bm m \cdot \bm \sigma \big),
\end{equation}
with $1-|\gamma|\leq \|\bm m\|$ guaranteed by positivity.
If $\gamma=0$ it is said to be unbiased, otherwise it is biased. A biased measurement can be decomposed as a mixture 
\begin{equation}\label{eq: biased decomp}
\{E_+,E_-\} = \|\bm m\| \{M_+, M_-\} + (1-\|\bm m\|) \{R_+,R_-\},
\end{equation}
of a PVM (unbiased) 
\begin{equation}
M_\pm= \frac{1}{2}( \mathbbm{1} \pm \hat{\bm m} \cdot \bm \sigma)
\end{equation}
with a unit vector $\hat{\bm m} = \frac{\bm m}{\|\bm m\|}$, and a "dummy measurement" $R_\pm=  q_\pm \mathbbm{1}$ with $q_\pm = \frac{1\pm \gamma-\|m\|}{2(1-\|m\|)}$
whose output is independent of the measured state. This observation leads to the following result.

\begin{lem}\label{Lemma:unbiased}
    Consider a channel $\Phi$ and a set of binary measurements $E_{\pm|y}= \frac{1}{2}\big((1\pm \gamma_y)\mathbbm{1} \pm \bm m_y \cdot \bm \sigma \big)$. The measurements set $\Phi^*(E_{\pm|y})$ is JM if the measurements $\Phi^*(M_{\pm|y})$ with $M_{\pm|y}= \frac{1}{2}\big(\mathbbm{1} \pm \hat{\bm m}_y \cdot \bm \sigma \big)$ and $\hat{\bm m}_y =\frac{\bm m_y}{\|\bm m_y\|}$ are JM.
\end{lem}
\begin{proof}
By Eq.~\eqref{eq: biased decomp} the measurements $\Phi^*(E_{\pm|y})$ can be decomposed as 
\begin{equation}
    \Phi^*(E_{\pm|y}) =  \|\bm m_y\| \Phi^*(M_{\pm|y}) +(1- \|\bm m_y\|) R_{\pm|y}
\end{equation}
where $R_{\pm |y} = \Phi^*(R_{\pm|y})= q_{\pm|y} \mathbbm{1}$ (the adjoint of a CPTP map is unital).

Since the measurements $\Phi^*(M_{\pm|y})$ are \JM{} by assumption there exists a parent POVM $E_\lambda$ simulating them. Then the measurements $\Phi^*(E_{\pm|y})$ are simulated by the following strategy. Perform the parent POVM $E_\lambda$. Upon receiving the setting $y$ simulate $\Phi^*(M_{\pm|y})$ with probability $\|\bm m_y\|$, or with probability $1-\|\bm m_y\|$ sample a random output from $q_{\pm|y}$.
\end{proof}

It implies that when discussing the \JM{} of sets of \emph{ideal} binary measurements subject to noise and loss, biased measurements can be ignored. That is, if we show that the set $\Phi^*_{\eta,v}(M_{\pm|y})$ with 
\begin{equation}\label{eq:qubit-meas-no-noise}
M_{\pm|y}=\frac{1}{2}(\mathbbm{1}\pm  \hat{\bm m}_y \, \cdot \bm \sigma)\,,
\end{equation}
and $\|\hat{\bm m}_y\|=1$ is \JM{} it automatically implies that any set of (biased) measurements $\Phi^*_{\eta,v}(E_{\pm|y})$ with $\bm m_y$ parallel to $\hat {\bm m}_y$ is also \JM. We will thus restrict our consideration to $N$ measurements of the form~\eqref{eq:qubit-meas-no-noise} in the rest of the section. Furthermore, we now set the channel $\Phi_v$  to be the white-noise channel $W_v$ mapping PVMs to 
\begin{equation}\label{eq: wn PVM}
W^*_v(M_{\pm|y}) = \frac{1}{2}(\mathbbm{1}\pm  v \, \hat{\bm m}_y \, \cdot \bm \sigma),
\end{equation}
which are also unbiased.

The question of the joint measurability of unbiased qubit measurements has been studied extensively, see e.g. section III.A in the review~\cite{RevModPhys.95.011003}.

In particular, from \cite{Busch86} it follows that any set of $N=2$ measurements of the form (\ref{eq: wn PVM}) is \JM{} if and only if
\begin{equation}\label{eq: v star 2}
  v \leq  v_*(\hat{\bm m}_1,\hat{\bm m}_2):=\frac{2}{ \|\hat{ \bm m}_1 + \hat{\bm m}_2\| + 
    \| \hat{\bm m}_1 - \hat{\bm m}_2\| }.
\end{equation} 

For $N=3$, the measurements are \JM{} if and only if~\cite{Pal_2011,yu2013quantum}
\begin{equation}\label{eq: v star 3}
   v \leq  v_*(\hat{\bm m}_1,\hat{\bm m}_2,\hat{\bm m}_3):= \frac{4}{ \sum_{k=0}^3 \| \bm t_{k} -\bm t_{FT}\|},
\end{equation}
where $\bm t_0 = \hat{\bm m}_1+\hat{\bm m}_2+\hat{\bm m}_3$, $\bm t_{k} = 2 \hat{\bm m}_k -\bm t_0 $ for $k\geq 1$, and $\bm t_{FT}$ is the so-called Fermat-Toricelli vector of the set $\{\hat{\bm m}_1, \hat{\bm m}_2,\hat{\bm m}_3\}$, i.e. $\bm t_{FT} = \text{argmin}_{\bm t } \sum_{k=1}^3 \|\hat{\bm m}_k - \bm t \|$

For any set of $N$ unbiased qubit measurements, we now show the following result.

\begin{lem}\label{Lemma:N qubit}
    A set of $N$ unbiased qubit measurements 
$M_{\pm|y}=\frac{1}{2}\left( \mathbbm{1} \pm \bm m_y\cdot \bm \sigma \right)$ specified by the vectors $\bm m_1,\dots, \bm m_N $  with $\|\bm m_y\|\leq 1$ is \JM{} if
\begin{align}\label{eq: Lemma N}
\sum_{\bm a}\| \sum_{k=1}^N (-1)^{a_k} \bm m_k\| \leq 2^N
\end{align}
where $\bm a=(a_1,\dots,a_N)\in\{0,1\}^N$ runs through the $2^N$ bitstrings. (See also the relaxed condition in Eq.~\eqref{eq: Lemma N relax}).
\end{lem}

\begin{proof}
Define the following set of $2^{N}$ vectors
\begin{equation}
    \bm w_{\bm a} = (-1)^{a_1} \bm m_1 + (-1)^{a_2} \bm m_2+\dots + (-1)^{a_N} \bm m_N 
\end{equation}
labeled by the bitstrings $\bm a = (a_1,\dots,a_N)\in\{0,1\}^N$ and denoted their normalized versions as $\hat{\bm w}_{\bm a}=\frac{\bm w_{\bm a}}{\|\bm w_{\bm a}\|}$, which define $2^N$ projective measurements. Using the probability distribution $p(\bm a)= \frac{\|\bm w_{\bm a}\|}{\sum_{\bm a}\|\bm w_{\bm a}\|}$ define the parent POVM composed of $2^{N+1}$ elements 
\begin{equation}
E_{\pm,\bm a} =  \frac{p(\bm a)}{2}(\mathbbm{1} \pm \hat{\bm w}_{\bm a} \cdot \bm \sigma) =\frac{1}{2}\left(
p(\bm a)\mathbbm{1}\pm \frac{\bm w_{\bm a} \cdot \bm \sigma}{ \sum_{\bm a}\|\bm w_{\bm a}\|}\right)\,.
\end{equation}
With the post-processing 
\begin{equation}
p(b|y,\pm,\bm a) = 
\begin{cases}
    1 & \pm  1= \text{sign}( b\, (-1)^{a_y}) \\
    0 & \text{otherwise}
\end{cases}
\end{equation}
for $b=\pm 1$, the parent POVMs can simulate all measurements of the form
\begin{align}
\tilde M_{b|y} &= \sum_{\pm,\bm a} E_{\pm,\bm a} \, p(b|y,\pm,\bm a)
= \sum_{\bm a} E_{\text{sign}(b(-1)^{a_y}),\bm a}
\\
&= \frac{1}{2}(\mathbbm{1} + b \frac{2^N}{\sum_{\bm a}\|\bm w_{\bm a}\|} \bm m_k \cdot \bm \sigma)\,.
\end{align}
These measurements are identical to $M_{b|y}$ if $\frac{2^N}{\sum_{\bm a}\|\bm w_{\bm a}\|}=1$, in which case the proof is done. If this ratio is even larger $\frac{2^N}{\sum_{\bm a}\|\bm w_{\bm a}\|}\geq 1$, one can mix some randomness into $\tilde M_{b|y}$ to complete the proof. 
\end{proof}

Applying this Lemma to the measurements $W_v^*(M_{\pm|y})$ in Eq.~\eqref{eq: wn PVM} with $\bm m_y =v\, \hat{\bm m}_y$ gives the follwing condition for their joint measurability
\begin{align}\label{eq: v star N}
v\leq v_\star(\hat{\bm m_1},\dots \hat{\bm m}_N):= \frac{2^N}{\sum_{\bm a}\| \sum_{k=1}^N (-1)^{a_k} \hat{\bm m}_k\|}
\end{align}
Note that Eq.~\eqref{eq: v star 2} is a particular case of Eq.~\eqref{eq: v star N} for $N=2$, while it is unclear to us if this is also the case for Eq.~\eqref{eq: v star 3}. An attentive reader also noticed that for general $N$ we used $v_\star$ instead of $v_*$ to indicate that the bound is potentially not tight.

The threshold visibilities for two, three and $N$ measurements in Eqs.~(\ref{eq: v star 2},\ref{eq: v star 3},\ref{eq: v star N}) can be used together with the Upper bound~\ref{ub:2'} to obtain sufficient conditions on the joint measurability of the corresponding CMUs $M_{b|y}^{\eta,v}$. The resulting bounds would depend on the specific choice of binary the qubit measurements described by the vectors $\hat{\bm m}_1,\dots,\hat{\bm m}_N$.

Nevertheless, it is also possible to derive a simple bound for any set of $N$  binary qubit measurements by noting that the rhs of Eq.~\eqref{eq: v star N} satisfies 
\begin{equation}\label{eq: v star N lb}
v_*(\hat{\bm m}_1,\dots,\hat{\bm m}_N) \geq \frac{1}{\sqrt{N}}.
\end{equation}
To see this first use the concavity of the square root to obtain $\sum_{\bm a} \frac{1}{2^N} \| \bm w_{\bm a} \|  = \sum_{\bm a}\frac{1}{2^N}\sqrt{\| \bm w_{\bm a}\|^2} \leq \sqrt{\frac{1}{2^N} \sum_{\bm a}\| \bm w_{\bm a}\|^2}$. Then in the last term simply compute 
$\sum_{\bm a}\| \sum_{k=1}^N (-1)^{a_k} \hat{\bm m}_k\|^2 = \sum_{\bm a} \sum_{kj=1}^N (-1)^{a_k+a_j} \hat{\bm m}_k \cdot \hat{\bm m}_j = \sum_{\bm a}\sum_{k=1}^N \hat{\bm m}_k \cdot \hat{\bm m}_k  = 2^N N $, which implies $\sqrt{\sum_{\bm a}\frac{1}{2^N}\| \sum_{k=1}^N (-1)^{a_y} \hat{\bm m}_k\|^2 } = \sqrt{ \frac{1}{2^N} 2^N N}  = \sqrt{N}$ and the inequality \eqref{eq: v star N lb}. This bound is, however, only saturable for $N=2$ and $3$ by the measurements corresponding to MUBs. Remarkably the same arguments imply 
$\frac{1}{2^N}\sum_{\bm a}\| \sum_{k=1}^N (-1)^{a_k} \bm m_k\|\leq \sqrt{\sum_k \|\bm m_k\|^2}$,
hence 
\begin{equation}\label{eq: Lemma N relax}
     \sum_{k=1}^N \|\bm m_k\|^2 \leq 1
\end{equation}
can be used as a relaxation of the condition (\ref{eq: Lemma N}) in  Lemma~\ref{Lemma:N qubit}.

To conclude this section we combine the Upper bound~\ref{ub:2'}, the Lemmas~\ref{Lemma:unbiased} with~\ref{Lemma:N qubit}, and the bound~\eqref{eq: v star N lb} to obtain the following simple sufficient condition for joint measurability of our CMUs.
\begin{ub}\label{ub:binary}
    The no-click CMU implementing $N$ binary qubit measurements (\ref{eq:wn-povms},\ref{eq:qubit-meas-no-noise}) with white-noise visibility $v$ and detection efficiency $\eta$ is \JM{} if
\begin{equation}\label{eq: N PVM simple}
    \eta \leq \frac{1}{\sqrt{N} \left((\sqrt{N}+1 )v-1\right)}.
\end{equation}
\end{ub}
It is easy to see that for $N \geq 4$ this condition does not bring any improvement over the general upper bound \eqref{eq:lb-noclick}, which reduces in the case $d=2$ to 
\begin{equation}\label{eq:lb-noclick-d2}
    \eta\leq \frac{2}{N\Big(3v-1\Big)}\,.
\end{equation}
In contrast, for $N=2$ and $3$, it improves on this bound, as illustrated in Fig.~\ref{fig:new}. We can also demonstrate numerically that in these cases the bound (\ref{eq: N PVM simple}) is tight. Indeed, selecting $\sigma_z$ and $\sigma_x$ as the two measurement directions in the case $N=2$, and additionally $\sigma_y$ in the case $N=3$, we can determine, for any fixed $v$, the maximal value of $\eta$ such that the resulting effective POVMs are \JM{} by solving a single SDP, as noted below Eq.~(\ref{eq:effective}). We observed that for all the tested values $v$, the maximal transmission $\eta$ found by the SDP matches with the upper bound of Eq.~\eqref{eq: N PVM simple} up to numerical precision.

\subsection{Bound on $\eta$ independent of $N$}
\label{sec:allmeas}

Finally, let us mention the asymptotic case where the number of measurements $N$ is unbounded. In \cite{sekatski2023compatibility} a necessary and sufficient condition for the joint measurability of all $d$-dimensional projective measurements subject to white noise and loss was derived. In the case of qubits, this condition takes a simple form
\begin{equation}\label{eq:allmeas}
    \eta\leq 2(1-v)\,,
\end{equation}
 depicted in Fig.~\ref{fig:new}. This bound improves over the upper bound (\ref{eq:lb-noclick-d2}) whenever $(1-v)(3v-1)>\frac{1}{N}$. By the Lemma~\ref{Lemma:unbiased} this result is promoted to all binary measurements.

For all POVMs, a sufficient condition is given by $\eta\leq (1-v)^{d-1}$ \cite{sekatski2023unlimited}.
When $N$ is large these bounds can be tighter and used instead of the bound (\ref{eq:lb-noclick}).

\section{Thermal-noise channel scenario}\label{sec:saa}

The bounds obtained in the previous section are based on the assumption that the measurements of the CMU have a finite number of discrete outcomes and that a photon loss can be modeled as a no-click outcome. However, some common quantum optics measurements are not of that form. Quadrature measurements are an example of measurements with an infinite, continuous set of outcomes, and which always produce an outcome, even when the photon is lost.

In this section, we thus make no hypothesis at all on the measurement device and we assume that losses in the channel are modeled, as is usual in quantum optics, through a thermal-noise (or attenuator) channel \cite{garcia2007quantum,serafini2017quantum} characterized through the \emph{transmittance} $\eta$ and the \emph{excess noise} $\epsilon$. 
Specifically, the thermal-noise channel can be viewed as combining on a beam-splitter of transmittance $\eta$, the incoming state $\rho$ and a thermal state $\tau_\nu$ with mean photon number $\nu=(\eta\epsilon)/(2(1-\eta))$, as illustrated in Fig.~\ref{fig:thnc}. The resulting CPTP map is 
\begin{equation}\label{eq:th-channel}
    \Phi_{\eta,\epsilon}(\rho) = \Tr_2\left(U_\eta \rho\otimes\tau_\nu U_\eta^\dagger\right)\,,
\end{equation}
where
\begin{equation}
    \tau_\nu = \frac{1}{\pi\nu}\int_{\mathbbm{C}} d^2\beta e^{-|\beta|^2/\nu}\ketbra{\beta}{\beta},\label{eq:ths}
\end{equation}
is a thermal state written in the basis of coherent states $\ket{\beta}$, and
\begin{equation}
    U_\eta=e^{\arccos(\sqrt{\eta}) (\hat{a}^\dagger \hat{b}-\hat{a}\hat{b}^\dagger)},\label{eq:bm_unitary}
\end{equation}
represent the action of the beam-splitter, where $\hat{a}$ and $\hat{b}$ are annihilation operators acting on the input state $\rho$ and the thermal state $\tau_\nu$, respectively.

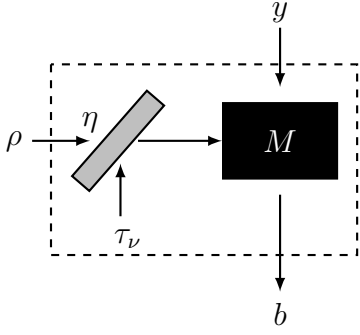
\begin{figure}[t]
\centering
    \begin{tikzpicture}[auto, thick, node distance=2.5cm, >=latex]
    
    \node [draw, fill=black, minimum width=1.5cm, minimum height=1cm] (box) {\large \textcolor{white}{$M$}};
    \node at (box.center) {};
    
    \node [above of=box, node distance=1cm] (input) {};
    \node [below of=box, node distance=1cm] (output) {};
    
    \node [left of=box, node distance=3.5cm] (rho) {\large $\rho$};
    
    \draw [dashed] ($(box.north west)+(-2.25,0.5)$) rectangle node {} ($(box.south east)+(0.25,-1)$);
    
    \draw [->] ($(input)+(0,0.5)$) -- ($(box.north)+(0,0.2)$);
    \draw [->] ($(box.south)+(0,-0.2)$) -- ($(output)-(0,1.)$);

    \node [left of=box, node distance=2.cm] (bm) {};
    \node [draw, fill=gray!50, rotate=49, minimum width=1.5cm, minimum height=0.3cm] at (bm.west) {};
    
    \draw [->] (rho) -- (-2.5,0);
    \draw [->] (bm.east) -- (box.west);
    \draw [->] ($(bm) - (0.1,1)$) -- ($(bm) - (0.1,0.3)$);
    
    \node [below of=bm, node distance=1.3cm] {\large $\tau_\nu$};
    \node [above of=box, node distance=1.7cm] {\large $y$};
    \node [below of=box, node distance=2.3cm] {\large $b$};
    \node at (-2.5,0.25) {\large $\eta$};
    \end{tikzpicture}
    \caption{CMU where the channel is a thermal-noise channel which can be viewed as combining $\rho$ with a thermal state $\tau_\nu$ at a beam-splitter of transmittance $\eta$.}\label{fig:thnc}
\end{figure}

\begin{ub}\label{ub:7}
    The channel $\Phi_{\eta,\epsilon}$ of a thermal-noise CMU is $N$-extendable, and thus the corresponding CMU is \JM{} irrespective of the measurements $M_y$, if
    \begin{equation}\label{eq:lb-thnc}
        \eta\leq\frac{1}{N(1-\epsilon /2)}\,.
    \end{equation}
\end{ub}
\begin{proof}
    Consider the channel $\Phi_{1\to N} = \Phi_{\text{BS}_N}\circ\Phi_\text{Amp}$ corresponding to first amplifying the input state $\rho$ with a gain $G>1$ and then splitting it using a symmetric $N$-mode beam-splitter, as illustrated in Fig.~\ref{fig:thnc-jm}. We now show that if we trace out $N-1$ output modes at the exit of this channel, we get the thermal-noise channel  $\Phi_{\eta,\epsilon}$ by choosing appropriately $G$. 

    The thermal-noise channel, the amplifier, and the beam-splitter are Gaussian channels, and their action can be fully described in terms of two real matrices, the scaling and noise matrices $X$ and $Y$ \cite{garcia2007quantum}. These matrices capture the evolution of the characteristic function of the system, regardless of whether the initial state is Gaussian or not. For the thermal-noise channel, the amplifier, and the channel corresponding to tracing out all but one ouptput mode of a symmetric $N$-mode beam-splitter, these matrices are respectively 
    \begin{equation}
        X_\text{Th} = \sqrt{\eta}\mathbbm{1},\qquad Y_\text{Th} = (1-\eta+\epsilon\eta)\mathbbm{1}\label{eq:simTh}\,,
    \end{equation}
    \begin{equation}
        X_{\text{Amp}}=\sqrt{G}\mathbbm{1}_2,\qquad Y_{\text{Amp}}=\left( G-1 \right)\mathbbm{1}_2,
    \end{equation} 
    and
    \begin{equation}
        X_{\text{BS}}=\sqrt{\frac{1}{N}}\mathbbm{1}_2,\qquad Y_{\text{BS}}= \frac{N-1}{N} \mathbbm{1}_2,.
    \end{equation}
    In the Gaussian systems framework, the channel obtained by applying two channels in succession, each described by the matrices $X_{A,B}$ and $Y_{A,B}$, is described by the matrices
    \begin{align}
    X_{B\circ A}&=X_{B}X_{A}, \\
    Y_{B\circ A}&=X_{B}^T Y_{A}X_{B}+Y_{B}.
    \end{align} 
    We thus have 
    \begin{align}
    X_{\text{BS}\circ\text{Amp}} &= X_{\text{BS}}X_{\text{Amp}}=\sqrt{\frac{G}{N}}\mathbbm{1}_2\,, \\
    Y_{\text{BS}\circ\text{Amp}} &= X_{\text{BS}}^T Y_{\text{Amp}}X_{\text{BS}}+Y_{\text{BS}}=\frac{G-N-2}{N}\mathbbm{1}_2\,.\nonumber
    \end{align} 
    These matrices are equal to the target matrices $X_\text{Th}$ and $Y_\text{Th}$ for $G=1/(1-\epsilon/2)$ and 
    \begin{equation}
        \eta=\frac{1}{N(1-\epsilon/2)}.
    \end{equation}
    Of course, lower values of $\eta$ are also possible by simply attenuating the overall channel output, resulting in the bound (\ref{eq:lb-thnc}).
\end{proof}
\begin{figure}[t]
	\centering
    \resizebox{.6\linewidth}{!}{
    \begin{tikzpicture}[auto, thick, node distance=2.5cm, >=latex]
    
    \node (laptop) {};
    \node [left of=laptop, node distance=4.3cm] (rho) {\large $\rho$};
    
    \node [draw, fill=black, minimum width=0.9cm, minimum height=4.cm, right of=rho, node distance=2.8cm] (meas) {\textcolor{white}{B.S.}};
    
    \node [draw, regular polygon, regular polygon sides=3, shape border rotate=90, fill=black, right of=rho, minimum height=0.3cm, node distance=1.5cm, scale=0.6] (amp) {\large \textcolor{white}{Amp}};
    
    \draw [->] (rho) -- (amp.west);
    \draw [->] (amp) -- (meas.west) ;
    
    \draw[->] (-1.,1.6) -- (-0,1.6) node[midway,above] (r1) {$\rho_1$};
    \draw[->] (-1.,0) -- (-0,0) node[midway,above] (r2) {$\rho_\mu$};
    \draw[->] (-1.,-1.6) -- (-0,-1.6) node[midway,above] (r3) {$\rho_N$};
    
    \draw[-,dotted] (r1.south) -- (r2.north);
    \draw[-,dotted] (r2.south) -- (r3.north);
    
    \end{tikzpicture}
    }
    \caption{Strategy for establishing the $N$-extendibility of the thermal-noise-channel. The initial state $\rho$ is amplified and split into $N$ modes using a beam splitter.\label{fig:thnc-jm}}
\end{figure}
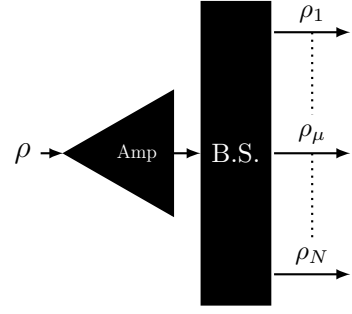

We note that the above result was also obtained in \cite{lami2019}, where a necessary and sufficient condition for the $N$-extendibility of single-mode Gaussian channels is presented. The above proof, however, provides an explicit strategy for establishing the $N$-extendibility of the thermal-noise channel.

While we considered a single-mode thermal-noise channel $\Phi_{\eta,\epsilon}$, the argument extends to $n$-mode scenarios, where the global noise channel is the product $\bar{\Phi}_{\eta,\epsilon}= \bigotimes_{i=1} \Phi_{(\eta,\epsilon)}$. It is easy to see that this channel $\bar{\Phi}_{\eta,\epsilon}$ is  $N$-extendable if $\Phi_{\eta,\epsilon}$ is. Therefore the Upper bound~\ref{ub:7} also holds in this scenario. 

The bound (\ref{eq:lb-thnc}) is completely generic as it depends only on the number of measurements $N$ and the excess noise $\epsilon$, and holds for any possible measurements $M_{y}$ implemented by the CMU. A much more stringent bound can be obtained if we assume that the measurements $M_y$ are Gaussian measurements, which can be implemented by first applying a Gaussian channel to the state and then performing homodyne detection \cite{kiukas2013informationally}. 
\begin{ub}[\cite{rahimi-keshari_verification_2021}]
    A CMU consisting of a thermal-noise channel followed by Gaussian measurements is \JM{} if 
    \begin{equation}\label{eq:lb-guassian}
        \eta\leq\frac{1}{2-\epsilon}\,.
    \end{equation}
\end{ub}
Remarkably, the above bound on the transmissivity $\eta$ is independent of the total number $N$ of measurements performed. It follows from the fact that the thermal-noise channel is incompatibility breaking for Gaussian measurements as shown in \cite{rahimi-keshari_verification_2021}. We provide a simple proof of this bound in Appendix~\ref{ap:har} in the case of pure homodyne measurements.

\section{\label{sec:qkd} Applications to QKD protocols}
Our approach allows us to put bounds on the performance of any QKD protocol involving a CMU, in particular within the DI or SDI approaches. Consider a protocol in which one of the parties, say Bob, generates their copy of the secret key by performing an untrusted measurement on a quantum state received from an untrusted channel, i.e., the key is obtained by post-processing the input $y$ and output $b$ of a CMU. Then no key can be extracted from the protocol if the effective POVMs of the CMU are \JM{}. This follows from the fact that in any CMU that is \JM{}, Alice and Bob cannot certify the presence of entanglement, which is a necessary condition for the security of QKD \cite{curty_entanglement_2004}. 

This result can also easily be directly proven. Let $B'=BY$ where $B$ and $Y$ are the random variables describing the input and output of Bob's measuring device, let $A$ denote Alice's registers used in the QKD protocol and $\mathcal{E}$ be Eve's classical register (potentially obtained by measuring a quantum system she holds).  It is known that for arbitrary post-processing of $A$ and $B'$, the key rate that can be extracted is upper-bounded by the intrinsic information~\cite{renner2003new}
\begin{equation}
    I(A:B'\downarrow \mathcal{E}) \equiv \min_{\mathcal{E}\to \mathcal{F}} I(A:B'|\mathcal{F}), 
\end{equation}
where $\mathcal{F}$ runs through all post-processing of Eve's register $\mathcal{E}$ and $I(A:B'|\mathcal{F})= \sum_f p(f) I(A:B'|\mathcal{F}=f) $ is the conditional mutual information.

Now consider the case where Bob's CMU is \JM{}. Then, as Eve controls the public channel of the CMU, we can assume that she implements herself the parent POVM $E$, obtains the output $\Lambda$, keeps a copy for her, and sends the other copy to Bob's apparatus that uses it to determine $B$ given $Y$. Hence, the probability distribution conditional on Eve's register $\mathcal{E} =\Lambda$ factorizes
\begin{equation}\label{eq: intrinsic conditional}
p(A=a,B'=by|\mathcal{E}=\lambda) = p(a|\lambda) p(y) p(b|y,\lambda)
\end{equation}
and $I(A:B'\downarrow \mathcal{E})\leq I(A:B'| \mathcal{E})=0$, proving that no key rate can be extracted.

We conclude from the above discussion that in any DI or SDI QKD protocol where Bob's device is untrusted and his inputs and outputs are used to establish the final key, no key can be extracted if the losses and noise satisfy the various bounds presented above.
However, these bounds can be further improved by using the concept of partial-joint-measurability, which we introduce in the next section.

\subsection{Partial-joint measurability}\label{sec:KJM}
In many QKD protocols, the key on Bob's side is obtained by post-processing the input $y$ and output $b$ only of a subset $\mathcal{K}\subset \{1,\ldots,N\}$ of cardinality $K=|\mathcal{K}|<N$ of all possible measurement inputs, while other measurements are solely used for parameter estimation. In this situation, we can weaken the notion of \JM{} through the following definition.

We say that the POVMs $M_y$ for $y\in\{1,\ldots,N\}$ are $\mathcal{K}$-JM, where $\mathcal{K}$ is subset of $\{1,\ldots,N\}$, if there exists \emph{i)}, a quantum instrument $I = \{I_\lambda\}_{\lambda}$ with classical measurement outcomes $\lambda$, \emph{ii)} probability distributions $p(b|y,\lambda)$, and \emph{iii)} measurement $M_{y,\lambda}$ such that, for any state $\rho$, the probabilities $p(b|y,\rho) = \Trace\left[\rho M_{b|y}\right]$ can be written as
\begin{equation}\label{eq: part JM}
    P(b|y,\rho) = 
    \begin{cases}
        \sum_{\lambda} p(b|y,\lambda)\Trace\left[I_\lambda(\rho)\right] & \text{if } y\in \mathcal{K}\\
        \sum_{\lambda} \Trace\left[I_\lambda(\rho) M_{b|y,\lambda}\right] & \text{if } y\notin \mathcal{K} .
    \end{cases}
\end{equation}
Operationally, this means that instead of carrying out the original POVMs $M_y$ on the state $\rho$, one can instead first perform the parent instrument $I$ on the state $\rho$, obtaining a classical output $\lambda$ and a quantum output $I_\lambda(\rho)$. If $y\in \mathcal{K}$, then one only uses the classical output $\lambda$ to generate the outcome $b$ with probability $p(b|y,\lambda)$, as in the case of full \JM{}. If $y\notin \mathcal{K}$, then one generates $b$ by carrying out on the state $I_\lambda(\rho)$ a measurement defined by the operators $\{M_{b|y,\lambda}\}$, which depend on $y$, but also on $\lambda$.

Since the relations \eqref{eq: part JM} should hold for any input state $\rho$, they can be equivalently formulated as
\begin{equation}
    M_{b|y} =
    \begin{cases}
        \sum_{\lambda} p(b|y,\lambda) I_\lambda^*(\mathbbm{1}) & \text{if } y\in \mathcal{K}\\
        \sum_{\lambda} I_\lambda^*( M_{b|y,\lambda}) & \text{if } y\notin \mathcal{K}
    \end{cases},
\end{equation}
where $I_\lambda^*$ denotes the adjoint of $I_\lambda$ and $I_\lambda^*(\mathbbm{1}) = E_\lambda$ defines a parent POVM as in the case of full \JM{}.

In the context of QKD applications, when a CMU is $\mathcal{K}$-JM, we can assume that the parent instrument $I_\lambda$ is performed by Eve, which can store a copy of its classical output $\lambda$. If Bob's register $B'=BY$ used to distill the final key only depends on the inputs $y\in \mathcal{K}$ (while the other inputs may still be used for parameter estimation) the intrinsic information is zero 
\begin{equation}
    I(A:B'\downarrow \mathcal{E}) = 0\,.
\end{equation}
This is because, analogously to the discussion leading to Eq.~\eqref{eq: intrinsic conditional}, the probability distribution of Alice's and Bob's random variable factorizes $p(A,B'|\mathcal{E}=\lambda)=p(A|\mathcal{E}=\lambda)p(B'|\mathcal{E}=\lambda)$ when conditioned on $\lambda$, which is held by Eve. Hence, if the CMU gives rise to $\mathcal{K}$-JM measurements, no key can be distilled from the outputs of any measurements in $\mathcal{K}$.

Note that in the case where \( K = N - 1 \), partial joint-measurability becomes equivalent to full joint-measurability. This is because, for \( K = N - 1 \), there is only a single measurement setting \( y \notin \mathcal{K} \). Therefore, there is no need to send any quantum state to Bob's device as this single measurement can immediately be performed after the parent instrument. The concatenation of the parent instrument and this subsequent measurement represents a parent POVM providing a classical outcome for all values of $y$, i.e., it implies that the set of Bob's POVM is fully JM.

In the case of the no-click scenario, the following result relates the full joint-measurability of a subset $\mathcal{K}$ of the measurements to the partial joint-measurability of the full set of measurements.
\begin{ub}
    Consider a no-click CMU with an arbitrary number $N$ of inputs and assume that for a subset $\mathcal{K}\subseteq\{1,\ldots,N\}$, the measurements $\{\Phi_\eta^*(M_y)\}_{y\in\mathcal{K}}$ are \JM{} for a detection efficiency $\eta$.
    Then the full CMU is $\mathcal{K}$-\JM{} for detection efficiency $\eta' \leq \eta/(1+\eta)$. 
\end{ub}
\begin{proof}
    By assumption, the set $\{\Phi_\eta^*(M_y)\}_{y\in\mathcal{K}}$ is \JM{}, which means that there exists a parent instrument $\mathcal{E}$ that simulates it. 
    Define then the following strategy. With probability $q$, perform the parent POVM $\mathcal{E}$ and follow the corresponding simulation scheme if $y\in \mathcal{K}$, while output no-click $\varnothing$ if $y \notin \mathcal{K}$. With probability $1-q$, leave the system untouched, output no-click $\varnothing$ if $y \in  \mathcal{K}$, and perform the  measurement $\Phi^*(M_{y})$ if $y \notin \mathcal{K}$. On average this strategy realizes the POVMs with elements
\begin{equation}
    \tilde M_{b|y}  = 
    \begin{cases} q \eta \Phi^*(M_{b|y}) & \text{if } b \neq \varnothing\\
        \left(q(1-\eta) + (1-q)\right) \mathbbm{1} & \text{if } b = \varnothing
    \end{cases}
\end{equation}
if $y\in \mathcal{K}$, and
\begin{equation}
    \tilde M_{b|y}  = 
    \begin{cases} (1-q) \Phi^*(M_{b|y}) & \text{if } b \neq \varnothing\\
        q \mathbbm{1} & \text{if } b = \varnothing
    \end{cases}
\end{equation}
if $y\notin \mathcal{K}$. These measurements have the form of the no-click CMU (\ref{eq:effective}) if $q\eta = (1-q) = \eta'$, which implies $\eta' = \eta/(1+\eta)$.
\end{proof}

As an illustration, if we apply the above construction to the bound (\ref{eq: N PVM simple}), we obtain that a white-noise no-click CMU implementing the qubit measurements (\ref{eq:qubit-meas-no-noise}) is $\mathcal{K}$-JM if
\begin{equation}\label{eq:lb-noclick-pjm}
    \eta \leq \frac{1}{v(K+ \sqrt{K})}\,,
\end{equation}
which improves the original bound when $K\leq N-1$. 

Note that by construction, the condition of partial joint measurability is weaker than that of joint measurability. The bound \eqref{eq:lb-noclick-pjm} allows us to illustrate this difference explicitly. 
Indeed, we mentioned below the bound \eqref{eq: N PVM simple} that for $N=3$, it is tight. In particular, the three MUBs measurements $\sigma_x,\sigma_y,\sigma_z$ are incompatible if $\eta > 2/(9v-3)$. However, from (\ref{eq:lb-noclick-pjm}) these three measurements are 1-\JM{} up to $\eta\leq 1/(2v)$. For instance,  for $v=1$, the three MUB measurements are incompatible but 1-\JM{} in the region $\frac{1}{3}<\eta\leq \frac{1}{2}$.

The construction above can also be applied to, e.g., the generic white-noise no-click bound (\ref{eq:lb-noclick}). However in this case a better improvement can be obtained using the following relation between partial joint measurability and channel extendibility.

\begin{lem}\label{Lemma:k-ext} If a channel $\Phi$ is  $(K+1)$-extendable, then the effective POVMs $\Phi^*(M_{y})$ are $\mathcal{K}$-JM for any subset $\mathcal{K}\subseteq\{1,\ldots,N\}$ of $K$ measurements. 
\end{lem}

To see this it is sufficient to apply, as illustrated in Fig. \ref{fig:partext}, the extended channel $\Phi_{1\to K+1}$, perform the intended measurements $M_y$ for $y$ in $\mathcal{K}$ on the first $K$ systems and leave the last one untouched. This defines a quantum instrument that satisfies all the required properties.
\begin{figure}[t]
    \centering
    \resizebox{.9\linewidth}{!}{
    \begin{tikzpicture}[auto, thick, node distance=2.5cm, >=latex]
    
    \node (laptop) {};
    
    \node [draw, minimum width=1.5cm, minimum height=0.8cm, above of=laptop, node distance=0.7cm] (laptop1) {};
    \node at (laptop1.center) {};
    \draw ([xshift=-0.62cm]laptop1.south) rectangle ([xshift=0.65cm,yshift=-0.5cm]laptop1.south);
    
    \draw ([xshift=-0.7cm,yshift=0.35cm]laptop1) rectangle ([xshift=0.7cm,yshift=-0.35cm]laptop1);
    
    \foreach \x in {0,...,5}
        \foreach \y in {0,...,2}
            \draw ([xshift=-0.55cm + \x * 0.2cm,yshift=-0.15cm - \y * 0.15cm]laptop1.south) rectangle ++(0.15, 0.1);
    
    \node [above of=laptop, node distance=1cm] (input) {};
    \node [below of=laptop, node distance=1cm] (output) {};
    
    \node [left of=laptop, node distance=7.5cm] (rho) {\large $\rho$};
    
    \draw [dashed] ($(laptop.north west)+(-6.5,2.3)$) rectangle node {} ($(laptop.south east)+(0.75,-2.3)$);
    
    \node [draw, fill=black, minimum width=1.2cm, minimum height=4.5cm, right of=rho, node distance=2.cm] (meas) {\textcolor{white}{$\Phi_{1\rightarrow K+1}$}};
    
    \node[draw, fill=black, minimum width=0.9cm, minimum height=0.7cm, anchor=center] (c) at (0,-1.85) {\textcolor{white}{$M_y$}};

    \node[draw, fill=black, minimum width=0.9cm, minimum height=0.7cm] (m1) at (-3.4,1.85) {\textcolor{white}{$M_1$}};
    \node[draw, fill=black, minimum width=0.9cm, minimum height=0.7cm] (m2) at (-3.4,0.5) {\textcolor{white}{$M_\mu$}};
    \node[draw, fill=black, minimum width=0.9cm, minimum height=0.7cm] (m3) at (-3.4,-0.85) {\textcolor{white}{$M_{K}$}};
    
    \node (l1) at (-2.3,1.85) {$\lambda_1$};
    \node (l2) at (-2.3,0.5) {$\lambda_\mu$};
    \node (l3) at (-2.3,-0.85) {$\lambda_{K}$};
    
    \draw [->] ($(input)+(0,1.6)$) -- ($(laptop1.north)+(0,0.2)$) node[midway,left] {if $y\in \mathcal{K}$};
    \draw [-] (0.2,2.8) -- (1,2.8);
    \draw [-] (1,2.8) -- (1,-0.5);
    \draw [-] (1,-0.5) -- (0,-0.5);
    \draw [->] (0,-0.5) -- (c) node[midway,left] {if $y\notin \mathcal{K}$};
    \draw [->] ($(c)+(0,-0.3)$) -- ($(output)-(0,1.65)$);
    \draw [->] (rho) -- (meas.west);

    \draw[->] (-4.9,1.85) -- (m1.west) node[midway,above] (r1) {$\rho_1$};
    \draw[->] (-4.9,0.5) -- (m2.west) node[midway,above] (r2) {$\rho_\mu$};
    \draw[->] (-4.9,-0.85) -- (m3.west) node[midway,above] (r3) {$\rho_{K}$};
    \draw[->] (-4.9,-1.85) -- (c) node[midway,above] (r4) {$\rho_{K+1}$};

    \draw[-] (m1.east) -- (l1.west);
    \draw[-] (m2.east) -- (l2.west);
    \draw[-] (m3.east) -- (l3.west);
    
    \draw[-,dotted] (m1.south) -- (m2.north);
    \draw[-,dotted] (m2.south) -- (m3.north);
    
    \draw[-,dotted] (r1.south) -- (r2.north);
    \draw[-,dotted] (r2.south) -- (r3.north);
    
    \draw[->] (l1.east) to [out=0,in=180] (laptop1.west);
    \draw[->] (l2.east) to [out=0,in=180] (laptop1.west);
    \draw[->] (l3.east) to [out=0,in=180] (laptop1.west);
    
    \node [above of=laptop, node distance=2.8cm] {\large $y$};
    \node [below of=laptop, node distance=2.9cm] {\large $b$};
    
    \end{tikzpicture}
    }
    \caption{Consider a CMU that admits an implementation with a channel $\Phi$ that is $K+1$-extendable. Perform each of the $K$ measurements $M_y$ for $y\in \mathcal{K}$ in the first $K$ extensions, resulting in outcomes $\lambda_1,\ldots,\lambda_{K}$, and forward the $K+1$th extension to the measurement device. Then by simply outputting $\lambda_y$ if $y\in \mathcal{K}$ and performing the measurement $M_y$ on the $K+1$th extension if $y\notin \mathcal{K}$ defines a CMU with effective POVMs that are $\mathcal{K}$-JM.}
    \label{fig:partext}
\end{figure}
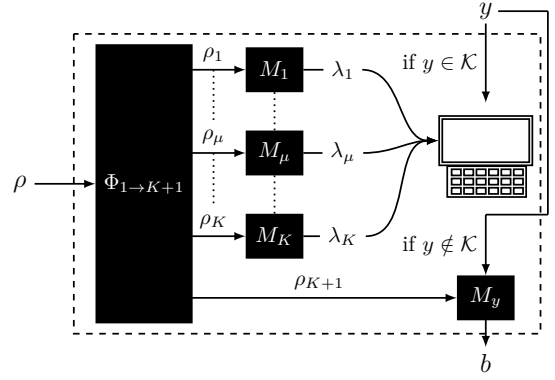

The above observation can be applied to the bounds (\ref{eq:lb-noclick}) and (\ref{eq:lb-thnc}) and implies that the white-noise no-click CMU is $\mathcal{K}$-JM if
\begin{equation}\label{eq:lb-noclick-pjm2}
    \eta \leq \frac{d}{(K+1)\left(v(d+1)-1\right)}\,,
\end{equation}
and that the thermal-noise CMU is $\mathcal{K}$-JM if
  \begin{equation}\label{eq:lb-thnc-pjm}
        \eta\leq \frac{1}{(K+1)(1-\epsilon /2)}\,,
    \end{equation}
which both improve (or equal) the original bounds when $K \leq N-1$.

In the context of DIQKD or semi-DI QKD where the key is generated from a subset $\mathcal{K}$ of the measurements, the bounds (\ref{eq:lb-noclick-pjm}), (\ref{eq:lb-noclick-pjm2}), and (\ref{eq:lb-thnc-pjm}) can be used to strengthen the critical levels of losses and noises beyond which no key can be generated whenever $K\leq N-1$. For instance, in the thermal-noise channel scenario, no key can be generated if $\eta< 1/(2-\epsilon)$ whenever the key is generated from a single measurement, independently of the total number $N$ of measurements used in the protocol, including those used for parameter estimation. While in the white-noise no-click scenario, we deduce from (\ref{eq:lb-noclick-pjm2}) that no key can be generated for $K=1$  if $\eta\leq 1/(3v-1)$, corresponding, e.g., to $\eta\leq 0.5235$ for $v=0.97$.

\subsection{Convex combination attacks for public communication from Bob to Alice}\label{sec:mixt}
As we explained above in any QKD protocol involving Alice and Bob and where Bob's quantum device is an untrusted CMU, no security can be guaranteed if the CMU is (partially)-\JM{}, since Eve can then have full information about Bob's output, i.e., Alice and Bob are decorrelated conditioned on Eve's knowledge. This is true independently of the protocol used to distill the shared key from Alice and Bob's generated data. However, as Eve's \JM{} attack targets Bob's device, one would intuitively expect protocols involving one-way public communication from Bob to Alice for key distillation to be more sensitive to such attacks. Indeed, for such protocols, the maximum achievable key rate is given by the Devetak-Winter formula \cite{devetak2005distillation,rennerrate}
\begin{equation}\label{eq:DW}
   r_\infty = H(B|E)-H(B|A)
\end{equation}
where $H(X|Y)$ is the von Neumann entropy of $X$ conditioned on $Y$. This key rate is zero as long as Eve's uncertainty $H(B|E)$ on Bob's symbols is lower than Alice's uncertainty $H(B|A)$, thus Eve's does not necessarily need to have full information about Bob's raw key for the protocol to be insecure.

This leads us to consider a class of \emph{convex combination attacks} directly inspired by a similar strategy proposed in \cite{kolodynski2020device} in the context of DIQKD\footnote{In the case of DIQKD, our attacks improve on the ones of \cite{kolodynski2020device} because they take into account not only losses but also white-noise. However, we can further improve these attacks by targeting both CMUs of a DIQKD protocol, see next subsection. We also remark that it is important to consider separately the case where a no-click outcome is kept as a distinct outcome in Alice and Bob's post-processing and the case where it is discarded. In \cite{kolodynski2020device} only the former case is analyzed. An example of a DIQKD protocol where binning of the no-click outcome leads to a positive key rate for a detection efficiency that is lower than the threshold computed in \cite{kolodynski2020device} can be found in \cite{Masini2022simplepractical}.}. 
In such attacks, Eve employs a mixture of two strategies: either, with probability $p$, she implements a lossy and noisy version of the expected measurements performed by Bob, in which case she exploits the fact that this implementation is (partially) \JM{} to get full information about Bob's outcome. Or with probability $1-p$, she replaces Bob's CMU with an ideal, loss-free, and noise-free CMU, in which case she has no information about Bob's outcome. This mixed strategy allows her to reproduce the actual CMU used by Bob beyond the admissible level of losses and noises derived in the previous sections, while at the same time providing her with enough information to make $ r_\infty = H(B|E)-H(B|A)=0$.

We consider such attacks only in the no-click scenario. Indeed, they cannot straightforwardly be applied to CMUs with a thermal-noise channel, since such channels are not closed under convex combinations. Specifically, we show the following in Appendix~\ref{app: gaussian decomp}.
\begin{lem}\label{Lemma:gauss-decomp}
    A thermal-noise channel $\Phi_{\eta,\epsilon}$ that is not $N$-extendable cannot be decomposed as 
\begin{equation}\label{eq: gauss dec}
    \Phi_{\eta,\epsilon} =p\, \Phi_G + (1-p) \Phi' \qquad (0<p<1)
\end{equation}
where $\Phi_G$ is a $N$-extendable Gaussian channel and $\Phi'$ is an arbitrary channel.
\end{lem}   

We prove this observation by showing that the Q-function of the channel $\Phi'= (\Phi_{\eta,\epsilon}-p\, \Phi_G)/(1-p)$ (for a coherent state input) can not be positive for such a decomposition.\\

In the following, we thus focus on QKD protocols with one-way communication from Bob to Alice and involving a white-noise no-click CMU represented by effective POVMs $\{M_y^{\eta,v}\} $ of the form (\ref{eq:wn-povms}).  We assume that this set of POVMs is neither \JM{}, nor $\mathcal{K}$-JM, where $\mathcal{K}$ is the set of inputs used to establish the raw key, since otherwise we have already shown that no key can be extracted. 

Consider now an attack, where, with probability $p$, Eve uses a strategy $S_*$ and, with probability $1-p$, a strategy $S_1$. In the former case, Bob's untrusted measurement device implements the POVMs $M_y^{\eta^*,v^*}$, where $\eta^*$ and $v^*$ are chosen such that these POVMs are $\mathcal{K}-\JM{}$ so that Eve has full information about Bob's outcome, i.e., $H(B|E,S_*)=0$. In the latter case, Bob's measurement device implements the ideal POVMs $M_{y}^{1,1}$, which may prevent Eve from having any non-trivial information about his outcome. Conservatively, we can upper-bound Eve's information by $H(B|E,S_1)\leq H_{1,1}(B)$, where $H_{1,1}(B)$ is the entropy of Bob's outcome, averaged over the inputs $y\in \mathcal{K}$, when Bob performs ideal POVMs $M_{y}^{1,1}$. 

Altogether Eve's information on Bob's outcome is then given by
\begin{align}\label{eq:HBE}
    H(B|E) &= p H(B|E,S_*) + (1-p) H(B|E,S_1)\nonumber \\
    &\leq (1-p) H_{1,1}(B)\,,
\end{align}
hence the key rate is upper-bounded by
\begin{equation}\label{eq: keyrate UB}
    r_\infty \leq (1-p) H_{1,1}(B) - H_{\eta,v}(B|A)\,,
\end{equation}
where $p$ remains to be determined, and where $H_{\eta,v}(B|A)$ is computed from the actual correlations between Alice and Bob's outcomes, characterized by an efficiency $\eta$ and visibility $v$. 

This attack must reproduce on average Bob's actual POVMs $M^{\eta,v}_{b|y}$. Thus, it must satisfy
\begin{equation}
   M^{\eta,v}_{b|y} = p  M^{\eta_*,v_*}_{b|y} + (1-p)M^{1,1}_{b|y},
\end{equation}
which is equivalent to the following system of equations
\begin{equation}\label{eq: system eq}
 \begin{aligned}
 p(1-\eta_* \, v_*)= 1-\eta v \\
  p\, (1-\eta_* ) = 1-\eta \,.
 \end{aligned} 
\end{equation}
We have seen in the previous sections, that the POVMs $M_y^{\eta^*,v^*}$ are $\mathcal{K}$-JM if $\eta^*$ satisfies a bound $\eta^* \leq \hat \eta_{N,K,d}(v^*)$ which depends on the visibility $v^*$, the number of measurements $N$ or $K$, the Hilbert space dimension $d$, and other properties of the measurements. Eve will pick the transmission value $\eta^*$ which saturates this bound
\begin{equation}
\eta_* = \hat \eta_{N,K,d}(v_*).\label{eq:simulated_transm}
\end{equation}
Inserting this value for $\eta_*$ in Eq.~\eqref{eq: system eq}, we can then determine $p$ (and $v_*$) as a function of Bob's actual CMU parameters $\eta$ and $v$, and the obtained value of $p$ can be used in (\ref{eq: keyrate UB}) to get a bound on the key rate. We now illustrate this using several of the bounds obtained in the previous sections.\\

\textbf{Arbitrary measurements in dimension $d$.} Using the bound (\ref{eq:lb-noclick-pjm2}) and setting $\eta_*=\frac{d}{(K+1)\left(v_*(d+1)-1\right)}$ to the value saturating this bound, by solving Eq.~\eqref{eq: system eq} we find that $p= p_{K,N,d}'$ with
\begin{equation}\label{eq: many d}\begin{split}
        p_{K,N,d}'(\eta,v) = &\frac{\eta(1-v)+d(1-\eta v)}{d}
        \\
        &\times \max\left\{\frac{K+1}{K},\frac{N}{N-1}\right\}\,.
\end{split}
\end{equation}
Here we used the fact that if all the $N$ measurements are used to extract the key one simply replaces $K+1$ with $N$.
An interesting special case is the one where the key is extracted from an extremely large number of measurements. A bound on the key rate independent of $N$ and $K$ can be obtained from Eq.~\eqref{eq: many d}, by taking the limit of large $K,N$. We find 
\begin{equation}
     \lim_{K,N\rightarrow \infty} p_{K,N,d}'(\eta,v) = 1-\frac{\eta(v(d+1)-1)}{d},\\
\end{equation}
and thus by Eq.~\eqref{eq: keyrate UB} the key rate $r_\infty=0$ is zero if 
$ \frac{\eta(v(d+1)-1)}{d}\leq \frac{H_{\eta,v}(B|A)}{H_{1,1}(B)}.$\\

\textbf{A few binary measurements in dimension 2.}
In the case $d=2$, the value (\ref{eq:  many d}) based on the generic bound (\ref{eq:lb-noclick-pjm2}) can be used. However, in the specific case of qubit measurement of the form~\eqref{eq:qubit-meas-no-noise} and $K=N=2$ or $K=N=3$, it is more advantageous to use the bound (\ref{eq:lb-noclick-d2}). 
Setting $\eta_*$ to saturate this bound we find that $p$ is given by
\begin{equation}\label{eq: p N qutits}
     p''_{N}(\eta,v)= \frac{N (1-\eta  v)+\eta  \sqrt{N} (1-v)}{N-1}\,.
\end{equation}
Note  that for any specific set of $N$ qubit measurements, one can of course use a tighter \JM{} condition based on Upper bound~\ref{ub:2'} and (\ref{eq: v star N}) to derive a better attack.\\

\textbf{All binary measurements in dimension 2.}
Finally, as discussed in section~\ref{sec:allmeas} the set of all qubit PVMs subject to noise and loss become \JM{} for $\eta \leq 2 (1-v)$. Setting $\eta_* = 2(1-v_*)$ in \ref{eq: system eq} we find $p$ to be equal to
\begin{equation}\label{eq: all qubit PVMS p}
    p'''(\eta,v) =  1-\eta v +\sqrt{\eta  (1-v) (2-\eta(1+  v))}\, .
\end{equation}\\

For the case of a qubit $(d=2)$ with projective measurements, which is the most studied in the literature, we have thus obtained three different expressions for $p$ leading to different bounds on the key rate in Eq.~\eqref{eq: keyrate UB}. In the following table, we compare these expressions and indicate which one gives a better attack depending on the values $N$ and $K$.
\begin{table}[H]
\centering
\begin{tabular}{|c || c |c |c |}
\hline
\diagbox{$K$}{$N$}   & 2 & 3 & $\geq$ 4 \\
\hline \hline
1 & Eq.~\eqref{eq: p N qutits} \cellcolor{lime} & Eq.~\eqref{eq: many d} \cellcolor{cyan}& Eq.~\eqref{eq: many d} \cellcolor{cyan}\\
\hline
2 & Eq.~\eqref{eq: p N qutits} \cellcolor{lime}&  Eq.~\eqref{eq: p N qutits} \cellcolor{lime} & Eq.~\eqref{eq: many d} \cellcolor{cyan}\\
\hline
3 & x \cellcolor{lightgray}& Eq.~\eqref{eq: p N qutits} \cellcolor{lime}& Eqs.~\eqref{eq: many d} or \eqref{eq: all qubit PVMS p} \cellcolor{yellow}\\
\hline
$\geq$ 4 & x \cellcolor{lightgray}& x \cellcolor{lightgray}& Eqs.~\eqref{eq: many d} or \eqref{eq: all qubit PVMS p} \cellcolor{yellow}\\
\hline
$\to \infty$ & x \cellcolor{lightgray}& x \cellcolor{lightgray}& Eq.~\eqref{eq: all qubit PVMS p} \cellcolor{pink}\\
\hline
\end{tabular}
\caption{ In the case of  $N$ binary measurements (ideally) on a qubit, $K$ of which are used for key extraction, the table gives the optimal attack parameter $p$ among the expressions $p_{K,N,2}'(\eta,v),p_{N}''(\eta,v)$ or $p'''(\eta,v)$ in Eqs.~(\ref{eq: many d} ,\ref{eq: p N qutits}, \ref{eq: all qubit PVMS p}) derived by using different JM bounds discussed in the previous section. In addition to the table, we find that for $v=1$ the maximal value is always $p=  (1-\eta) \max\left \{\frac{K+1}{K}, \frac{N}{N-1}\right\}$.}\label{tab:bounds}
\end{table}

\subsubsection{Application to BB84/CHSH-type protocols}
We now illustrate the above approach on a family of qubit protocols where the key on Bob's side is based on the output $b\in\{0,1\}$ of a single measurement, which we assume in the honest implementation to be a $\sigma_z$ measurement. We also assume that in the key generation rounds, the state entering Bob's CMU is the $\sigma_z$ eigenstate $|0\rangle$ or $|1\rangle$, depending on Alice's key variable $a=0$ or $a=1$.
These are the only features of the protocol that we need to apply a bound on the key rate using our approach. 

Protocols satisfying the above conditions are the one-sided-DI version of the original BB84 protocol, where the source is fully characterized \cite{tomamichel2013one}, semi-DI versions of it, where Alice's device is only trusted to prepare qubit states \cite{woodhead2016semi}, the semi-DI CHSH prepare-and-measure protocol of \cite{woodhead2015secrecy}, entanglement-based protocols such as the one-sided-DI protocols based on steering of \cite{branciard2012one,masini-onesided}, fully DI CHSH-based protocols \cite{ref:ab2007}, or their routed version \cite{ref:routed}.

As $K=1$ and $d=2$, we have from (\ref{eq: many d}), $p = 2+\eta(1-3v)$, while $H_{1,1}(B)=1$. To compute $H_\eta(B|A)$, we can consider two possibilities. Either non-detection events are treated as distinct outcomes, i.e., $b\in\{0,1,\varnothing\}$. Or they are grouped with one of the other possible outcomes. In the first case, we have $H_{\eta,v}(B|A)\geq \eta\,h\left( \frac{1+v}{2}\right)+h(\eta)$, so that the key rate is upper-bounded as 
\begin{equation}
    r_\infty \leq \eta(3v-1)-1-\eta\,h\left( \frac{1+v}{2}\right)-h(\eta)\,.
\end{equation}
In the case where Bob bins his no-click outcome with one of the other outcomes, we have 
$H(B|A)\geq \frac{1}{2}\left(h\left(\frac{\eta}{2}(1+v)\right)+h\left(\frac{\eta}{2}(1-v)\right)\right)$ and the key rate is upper-bounded as
\begin{equation}
    r_\infty \leq \eta(3v-1)-1-\frac{1}{2}\left(h\left(\frac{\eta}{2}(1+v)\right)+h\left(\frac{\eta}{2}(1-v)\right)\right)\,.\nonumber
\end{equation}

As an illustration, for a visibility $v=1$, we find thresholds of $\eta=83.0\%$ in the case where Bob does not bin and $\eta=71.6\%$ in the case where he bins, and for a visibility $v=0.97$, we find thresholds of $\eta=87\%$ and $\eta=76\%$, respectively.

\subsubsection{Application to receiver-device-independent protocols}
As a second application, we consider the class of so-called receiver device-independent (RDI) protocols \cite{ioannou2022receiver,Ioannou2022njp}, which can be seen as an extension of the B92 protocol \cite{b92} to the case of many states and measurements. These are prepare-and-measure protocols as outlined in Fig.~\ref{fig:exp}(a). Here, the measuring device is completely untrusted, while the preparing device is assumed to produce $N$ pure states whose Gram matrix is fixed. 
We now provide a brief description of the protocol with an intuitive motivation for the ideal case (no noise), an interested reader is invited to consult the original papers for the full details.
We will consider the version of the protocol where Alice prepares qubit states 
\begin{equation}\label{eq: RDI states}
    \ket{\psi_x}=\cos(\theta/2)\ket{0}+e^{\frac{i2\pi}{N}x}\sin(\theta/2)\ket{1},
\end{equation}
for some choice of $\theta \in [0,\frac{\pi}{2}]$ and a value $x\in\{0,\dots,N-1\}$ that she choses at random, analyzed in \cite{Ioannou2022njp} in detail. Bob's ideal measurements are then $M_{b|y}=\ketbra{\psi_y^\perp},\ketbra{\psi_y}$ for $b=0,1$ respectively, with the setting $y\in \{0,\dots,N-1\}$ chosen at random.  After his measurement, Bob tells Alice to reject the round if $b\neq 0$. At this point, since $\braket{\psi_y}{\psi_y^\perp}=0$ both parties know that their values $x$ and $y$ must satisfy $x\neq y$  in all the rounds which are not rejected. Next, Alice reveals an unordered pair of values $(x,x')$ where $x$ is the value she encoded in the state and $x'$ is a different value chosen at random. This pair is only accepted by Bob if $y\in \{x,x'\}$. Since there are only two possible values left, and Bob can exclude one of them (from $ x\neq y$), Alice and Bob are left with a perfectly correlated bit after the sifting. 
This gives the idea of the protocol. We now explicitly compute the correlations observed by Alice and Bob. 

The probability to observe the outcome $b=0$ depends on the state $\ket{\psi_x}$ sent by Alice and is given by
\begin{align}\label{eq: RDI p 0}
    p^{(\eta,v)}(0|x,y) &= \eta\,  v |\braket{\psi_x}{\psi_y^\perp}|^2 + \frac{\eta (1-v)}{2} 
\end{align}
where $|\braket{\psi_x}{\psi_y^\perp}|^2= \sin^2(\theta) \sin^2(\frac{\pi (x-y)}{N})$.
From Bob's perspective (knowing $y$), a round is successful if $b=0$, and if $x=y$ (happening with probability $\frac{1}{N}$) or if $x'=y$ (for some $x\neq y$), summing the two gives
\begin{equation}
    p(\text{succ}|y) =  \frac{1}{N} p(0|y,y) +  \frac{1}{N(N-1)}\sum_{x\neq y} p(0|x,y).
\end{equation}
Combining with Eq.~\eqref{eq: RDI p 0} we find that the success probability is independent of $y$ and given by 
 \begin{align}\label{eq: p succ}
     p^{(\eta,v)}(\text{succ}) = \frac{\eta v \sin^2(\theta)}{2 (N-1)}+\frac{\eta(1-v)}{ N} 
 \end{align}

 Eve's strategy consists of replacing the CMU with a jointly measurable one, given by the parameters $(\eta_*,v_*)$, with probability $p=p(\text{attack})$, or with the ideal one with probability $1-p$. In fact, the probability that the round is successful depends on her choice.  The conditional probability of the attack reads
 \begin{equation}
   \bar p = p(\text{attack}|\text{succ}) =  \frac{p^{(\eta_*,v_*)}(\text{succ})\, p}{p^{(\eta,v)}(\text{succ})},
 \end{equation}
and implies $H(B|E,\text{succ})=(1-\bar p)$.
With the help of Eqs.~\eqref{eq: system eq} we see that 
\begin{align}
    p^{(\eta_*,v_*)}(\text{succ})\, p &= \frac{ p\eta_* v_* \sin^2(\theta)}{2 (N-1)}+\frac{p \eta_*(1-v_*)}{ N} \\
    & =  p^{(\eta,v)}(\text{succ}) - \frac{(1-p)\sin^2(\theta)}{2 (N-1)}
\end{align}
 and thus
     $\bar p = 1-(1-p)\frac{\frac{\sin^2(\theta)}{2 (N-1)}}{{p^{(\eta,v)}(\text{succ})}}$\,.

When Eve replaces, with probability $p$, the CMU with a \JM{} one, her parent POVM predicts, according to the strategies detailed in Section~\ref{sec:DV} that a no-click outcome $b=\{0,1\}$ will be obtained for Bob for only one of his possible measurements $y$. This implies that when Bob announces that he obtained the value $b=0$, Eve can infer which input $y$ was used by Bob and hence her entropy about Bob's raw key is 0. Analogously to Eq. (\ref{eq:HBE}), we thus have $H(B|E,\text{succ})\leq (1-\bar p)H_{1,1}(B|\text{succ})\leq 1-\bar p$ or
\begin{equation}
 H(B|E,\text{succ}) \leq (1-p)\frac{\frac{\sin^2(\theta)}{2 (N-1)}}{{p^{(\eta,v)}(\text{succ})}}.
 \end{equation}
While this formula is affected by post-selection, we see that the optimal strategy for Eve is still the one maximizing $p$ such that the noisy CMU is jointly measurable and we can use the formulas derived in the last section and summarized in Table.~\ref{tab:bounds} (in our case $K=N$ as all of Bob's measurements are used for the key).

Finally, to compute the key rate, we also need the expression of $H(B|A,\text{succ})$. For a successful round where Alice prepared the state $x$ and also revealed $x'$,
Bob must have chosen the measurements with $y=x$ or $y=x'$. For Alice the likelihood of these events is
\begin{align}
    p(y=x|A,\text{succ}) &= \frac{p^{(\eta,v)}(0|x,x)}{p^{(\eta,v)}(0|x,x)+p^{(\eta,v)}(0|x,x')} \\
    p(y=n|A,\text{succ}) &= \frac{p^{(\eta,v)}(0|x,x')}{p^{(\eta,v)}(0|x,x)+p^{(\eta,v)}(0|x,x')}
\end{align}
The entropy of this distribution (as well as the success probability) depends on the pair $\{x,x' \}$. The entropy is minimal if $\sin^2(\pi \frac{x-x'}{N})=1$ in which case it is equal to the rhs of
\begin{equation}
H(B|A,\text{succ}) \geq h\left(\frac{1-v}{2( 1-v \cos^2(\theta))} \right).
\end{equation}
This gives a general upper bound on the key rate with one-way error correction
\begin{align}
    r &\leq p(\text{succ})\Big( H(B|E,\text{succ}) - H(B|A,\text{succ}) \Big) 
    \\  &\leq 
    (1-p)\frac{\sin^2(\theta)}{2 (N-1)} - p^{(\eta,v)}(\text{succ}) h\left(\frac{1-v}{2(1-v \cos^2(\theta))} \right)\nonumber
\end{align}

To apply the bound we now take the limit of many measurements $N\to \infty$, where it was shown~\cite{ioannou2022receiver} that the key rate remains positive for arbitrary losses provided that the visibility is perfect $v=1$ and $N>\frac{1}{\eta}$. Since the ideal measurements of Bob are qubit PVMs we can use the Eq.~\eqref{eq: all qubit PVMS p} for the strategy of Eve that maximizes $(1-p)$. We find that for $\eta= 10 \%,1\%,0.1\%$ the key rate is zero if $v\leq 95\%, 99.5\%, 99.95\%$ respectively.

With the same approach tighter bounds can be derived on the RDI protocol by computing the entropy $H(B|A,\text{succ})$ explicitly, and deriving a tighter \JM{} condition for the specific set of the CMUs stemming from the measurements $M_{b|y}=\{\ketbra{\psi_y}, \ketbra{\psi_y^\perp}\}$ specified by Eq.~\eqref{eq: RDI states}.

\subsection{Convex combination attacks for DIQKD protocols}
\label{sec: full DI}
The convex combination attacks introduced in the previous section are asymmetric as they apply to generic QKD protocols where Bob holds a CMU, but where Alice may hold, e.g., a trusted preparation device. Fully DIQKD protocols, however, are symmetric in the sense that both parties hold a CMU. Exploiting this feature, we analyze in this section improved convex combination attacks for DIQKD protocols that lead to more stringent constraints on the key rate. Furthermore, again due to the symmetry of the protocol, they can be applied to situations where the public communication for key distillation goes from Bob to Alice, as in the previous section, or from Alice to Bob. For simplicity, we consider explicitly only the former case. But we note that in many cases the measurements and post-processing performed by Alice and Bob in DIQKD are essentially identical (up to e.g. local unitaries), implying that the same bounds on the key rate would be obtained by considering the public communication in the other direction.

Our convex combination attacks, in their simplest version, are a mixture of strategies where Eve replaces either Alice's or Bob's CMU with a full \JM{} for a portion of the time, and allowing the devices to operate ideally for the remainder. In the former case, the correlations between Alice and Bob are local, while they are non-local in the latter case. This class of attacks thus fits in the same framework as those proposed in \cite{farkas2021bell,lukanowski2022upper} involving a convex combination of local and non-local correlations. Such local/non-local mixtures are obtained in \cite{farkas2021bell,lukanowski2022upper} through analytical or numerical (linear programming) results tailored to specific quantum correlations. However, the analytical results are limited to highly specific cases and linear programming methods become increasingly challenging when dealing with a large number of measurements and outputs. In contrast, the approach based on joint measurability is generic, as it depends only on rough properties of the quantum protocol, and it can provide bounds even in scenarios with many measurements. The downside is that it may lead to bounds on the key rate less stringent than the ones of \cite{farkas2021bell,lukanowski2022upper} when applied to the specific correlations for which their methods are tailored.

Note, however, that more generally, we will consider attacks where Eve replaces Bob's CMU with a $\mathcal{K}$-\JM{} one with some probability. This gives rise to a class of attacks different from the one considered in \cite{farkas2021bell,lukanowski2022upper}, since correlations arising from a $\mathcal{K}$-\JM{} POVMs are not necessarily local, yet, are sufficient for Eve to have full information about Bob's key-generating outcomes. This may compensate for the fact that our joint measurability bounds are not tailored to specific correlations. 
In contrast, it is important to note that Alice’s device cannot be replaced with a $\mathcal{K}$-\JM{} one, as this does not allow Eve to guess Bob's key generating outcomes in general.

Let us now describe our attacks in more detail. Consider a generic DIQKD scenario consisting of a central source distributing bipartite states $\rho$, a CMU for Alice, and a CMU for Bob, as depicted in Fig.~\ref{fig:exp}.b. For concreteness, we assume that both Alice and Bob have a white-noise no-click CMU characterized by the effective POVMs (\ref{eq:wn-povms}) with detector efficiency $\eta$ and white-noise visibility $v$. Therefore, together they implement the joint effective POVMs $M_{a|x}^{\eta,v}\otimes M_{b|y}^{\eta,v}$ on the incoming state $\rho$. We further denote by $N_A$ the number of measurements on Alice's side, and by $N_B$ the overall number of measurements on Bob's side. Furthermore, let us assume that the key is only extracted from a subset of Bob's measurements with the inputs $y \in \mathcal{K}_B$.

Eve's attack is based on the following convex decomposition of Alice and Bob's joint POVMs
\begin{align}
    M_{a|x}^{\eta,v}\otimes M_{b|y}^{\eta,v} &= p_A M_{a|x}^{\eta_A,v_A}\otimes M_{b|y}^{1,1} + p_B M_{a|x}^{1,1}\otimes M_{b|y}^{\eta_B,v_B}\nonumber \\
    & + q M_{a|x}^{\eta',0}\otimes M_{b|y}^{\eta',0}\label{eq:convexattack}\\
    & + (1-p_A-p_B-q) M_{a|x}^{1,1}\otimes M_{b|y}^{1,1}\nonumber\,,
\end{align}
where $\eta' = \eta(1-v)/(1-\eta v)$, $q=(1-\eta v)^2$, and $\eta_A$, $v_A$, $\eta_B$, $v_B$, $p_A$, and $p_B$ are free parameters that Eve can chose. Using the explicit form for the white-noise no-click CMU POVMs (\ref{eq:wn-povms}), it is readily verified that this decomposition is valid provided that these parameters satisfy the following system of equations
\begin{align}
    2 \eta v(1-\eta v) &= p_A(1- \eta_A v_A) +p_B(1-\eta_B v_B) \label{eq:constr1}\\
    \eta^2 (1-v)v & =\eta _{A}p_{A} \left(1-v_{A}\right) =\eta _{B}p_{B} \left(1-v_{B}\right)  \\
    \eta  v  (1-\eta ) & = p_{A}\left(1-\eta _{A}\right) =  p_{B}\left(1-\eta _{B}\right)\,.
\end{align}
The first constraint (\ref{eq:constr1}) is linearly redundant and can be omitted. The remaining ones are equivalent to 
\begin{align}
    \eta v(1-\eta v) & =p_{A} \left(1-\eta_A v_{A}\right) =p_{B} \left(1-\eta_B v_{B}\right) \label{eq:att1} \\
    \eta  v (1-\eta )  & = p_{A}\left(1-\eta _{A}\right) = p_{B}\left(1-\eta _{B}\right) \,.   \label{eq:att2}
\end{align}
Here, one sees that the tuples of parameters $(p_A,\eta_A,v_A)$ and $(p_B,\eta_A,v_A)$ are subject to independent constraints, leaving one free parameter for each tuple.

The convex decomposition (\ref{eq:convexattack}) can be exploited by Eve as follows. With probability  $1-p_A-p_B-q$, Alice's and Bob's CMUs behave ideally as they implement the ideal POVMs $M_{a|x}^{1,1}\otimes M_{b|y}^{1,1}$. In this case, we assume that Eve has no extra information about Bob's outcomes, i.e. her conditional entropy is only bounded by 
$H(B|E)\leq H_{1,1}(B)$,
where $H_{1,1}(B)$ the entropy of Bob's outcomes averaged over the inputs $y\in\mathcal{K}_B$ in the ideal (noise-free and loss-free) situation. With probability $q$, the CMUs of Alice and Bob implement the POVMs $M_{a|x}^{\eta',0}\otimes M_{b|y}^{\eta',0}$, which are maximally noisy, i.e., the visibility is zero. In this case, the correlations generated by Alice and Bob are local. Since Eve can further decompose these correlations as a mixture of deterministic correlations, she has full information about Bob's outcomes, i.e., $H(B|E)=0$. With probability $p_A$, Alice's CMU implements the POVMs $M_{a|x}^{\eta_A,v_A}$. If we choose $\eta_A$ and $v_A$ such that these POVMs are (full) \JM{}, then, again,  Alice and Bob's correlations are local, and thus $H(B|E)=0$. Finally, with probability $p_B$, Bob's CMU implements the POVMs $M_{b|y}^{\eta_B,v_B}$. If we now choose $\eta_B$ and $v_B$ such that these POVMs are $\mathcal{K}_B$-\JM{}, Eve has again full information about Bob's outcomes for the inputs $y\in\mathcal{K}_B$, i.e., $H(B|E)=0$, since she can implement the parent instrument and predict all these outcomes. Note that in this last case, though, the correlations between Alice and Bob, which involve also the inputs $y\notin\mathcal{K}$, are not necessarily local. Summarizing, by our attack the key rate is upper-bounded as
\begin{equation}\label{eq: DIQKD bound}\begin{split}
    &r_\infty \leq t H_{1,1}(B) - H_{\eta,v}(B|A)\,,\\
    &\quad \text{with}\\
    &t = 1-q-p_A-p_B = 2\eta v -\eta^2v^2 -p_A-p_B\,.
    \end{split}
\end{equation}
Hence to find the tightest upper-bound, we must maximize the probabilities $p_A$ and $p_B$ subject to the constraints that the POVMs $M_{x}^{\eta_A,v_A}$ are full \JM{} and the POVMs $M_{y}^{\eta_B,v_B}$ are $\mathcal{K}$-\JM{}. For this, it is optimal for Eve to set the efficiency $\eta_A = \hat \eta_{N_A,d}(v_A)$ saturating one of the bounds for \JM{} that we have derived in the previous sections and similarly setting $\eta_B = \hat \eta_{N_B,K_B,d}(v_B)$ saturating one of the bounds for $\mathcal{K}_B$-\JM{}. Then, the optimal values of $p_A$ and $p_B$ are given by the solutions of the Eqs.~(\ref{eq:att1},\ref{eq:att2}), i.e 
\begin{equation}
    \begin{aligned}
    \hspace{-0.9em}\frac{p_{A}}{\eta v} \left(1-\hat \eta_{N_A,d}(v_A) \, v_{A}\right)  &=  \frac{p_B}{\eta v} \left(1-\hat \eta_{N_B,K_B,d}(v_B)v_{B}\right)\\
    &= 1-\eta v  \\
     \frac{p_{A}}{\eta v}\left(1-\hat \eta_{N_A,d}(v_A)\right) &=  \frac{p_{B}}{\eta v}\left(1-\hat \eta_{N_B,K_B,d}(v_B)\right)\\
     &=   1-\eta .
    \end{aligned}
\end{equation}
Remarkably, these equations become equivalent to those in \eqref{eq: system eq} by the following parameter change $ (p_{A(B)}, \eta_{A(B)}, v_{A(B)})\to(\eta v\, p,\eta_*,v_*)$. It follows that the expressions $ \eta \,v \, p_{K,N,d}'(\eta,v), \eta \,v \,  p_{N}''(\eta,v)$ or $\eta \,v \,  p'''(\eta,v)$ in Eqs.~(\ref{eq: many d} ,\ref{eq: p N qutits}, \ref{eq: all qubit PVMS p}) can be readily used to find $p_{A(B)}$ in different cases. 

In particular, our bound~\eqref{eq: DIQKD bound} guarantees that the key rate is zero if one can choose the values $p_A$ and $p_B$ such that
\begin{equation} \label{eq: no key DIQKD 1}
    p_A+p_B \geq 2 \eta v -\eta^2 v^2 - \frac{H_{\eta,v}(B|A)}{H_{1,1}(B)}
    \implies r_\infty=0.
\end{equation}
Noticing that the term $\frac{H_{\eta,v}(B|A)}{H_{1,1}(B)}$ is always positive one obtains a weaker but correlation-independent condition for zero key rate
\begin{equation}\label{eq: no key DIQKD 2}
    p_A+p_B \geq 2 \eta v -\eta^2 v^2
    \implies r_\infty=0.
\end{equation}
We now illustrate this bound on the three specific attack strategies already discussed in Sec.~\ref{sec:mixt}.\\

\textbf{Arbitrary measurements in dimension $d$.} When we make no assumption on the measurements performed except their dimensionality, we can use bounds (\ref{eq:lb-noclick}) and (\ref{eq:lb-noclick-pjm2}), and we obtain that the probabilities $p_A$ and $p_B$ must be
\begin{align}\label{eq:bnd DIQKD d A}
    p_A&= \eta\, v \, p_{N_A-1,N_A,d}'(\eta,v)\\
    \label{eq:bnd DIQKD d B}
    p_B & =\eta \,v \, p_{K_B,N_B,d}'(\eta,v)\,,    
\end{align}
with the function $p_{K,N,d}'(\eta,v)$ defined in equation \eqref{eq: many d}. In particular, no key can be extracted (Eq.~\ref{eq: no key DIQKD 2}), whenever
\begin{equation}
    \eta\leq\frac{\left(T_{N_A,K_B,N_B}-2\right) d}{\left(T_{N_A,K_B,N_B}-1\right) d \, v-T_{N_A,K_B,N_B}(1-v)},\nonumber
\end{equation}
where we introduced $T_{N_A,K_B,N_B}=\frac{N_A}{N_A-1}+\min\left\{\frac{K_B+1}{K_B},\frac{N_B}{N_B-1}\right\}$ to shorten the equation. This bound becomes $\eta \leq 1- \frac{1}{T_{N_A,K_B,N_B}-1}$ for $v=1$, and  $v\leq \frac{T_{N_A,K_B,N_B} + d(T_{N_A,K_B,N_B}-2)}{T_{N_A,K_B,N_B}+ d(T_{N_A,K_B,N_B}-1)}$ for $\eta=1$.\\


\textbf{A few binary measurements on a qubit.} When Alice and Bob use $N_A$, $N_B$ (respectively) binary qubit measurements we can use the bound in Eq.~\eqref{eq: N PVM simple} to obtain
\begin{equation}\label{eq:2-sides-2pvm}
    p_{A(B)}=\eta \, v \, p''_{N_{A(B)}}(\eta, v),
\end{equation}
with the function $p''_{N}(\eta, v)$ defined in Eq.~\eqref{eq: p N qutits}.
In particular, for $N_A=N_B=N$ we find that no key can be extracted (Eq.~\eqref{eq: no key DIQKD 2}) whenever
\begin{equation}\nonumber
    \eta\leq\frac{2(N+1)v+4(1-v)\sqrt{N}}{(N-1)^2v^2+4N(2v-1)},
\end{equation} 
which becomes $v\leq \frac{2}{\sqrt{N}+1}$ for $\eta=1$.\\

\textbf{All binary measurements on a qubit.} Finally, in the same situation but with an unrestricted number of measurements for Alice and Bob, we can use the bound of Eq.~\eqref{eq:allmeas} valid for the set of all qubit PVMs to obtain
\begin{equation}\label{eq:2sides-all}
    p_{A(B)}=\eta \, v \, p'''(\eta,v),
\end{equation} 
with $p'''(\eta,v)$ defined in Eq.~\eqref{eq: all qubit PVMS p}.
In particular, we find that no key can be extracted (Eq.~\eqref{eq: no key DIQKD 2}) if
\begin{equation}\nonumber
\eta \leq \frac{8 (1-v)}{4-3 v^2},
\end{equation}
which becomes $v\leq \frac{2}{3}$ for $\eta=1$.\\

Note that different bounds can be used on the sides of Alice and Bob. Furthermore, the attack we presented can be easily improved with additional knowledge about the measurements performed by the parties, as in this case one can use tighter (partial)-JM bounds specific to a given CMU. 

Recall also that in the most common case where Alice and Bob perform qubit PVMs, the comparison of the three bounds $p'_{K,N,2}(\eta,v)$, $p''_{N}(\eta,v)$ and $p'''(\eta,v)$ was given in Table~\ref{tab:bounds}. It can be readily used to find the optimal attack of Eve in different cases. With its help, let us now explore the performance of popular DIQKD protocols in noisy conditions and determine the maximum achievable key rate. The challenge here is to compute the terms $H_{\eta,v}(B|A)$ and $H_{1,1}(B)$.

\subsubsection{Application to qubit protocols without binning}

Consider the protocols where Alice and Bob share a two-qubit state
\begin{equation}
\ket{\Psi_\theta}=\cos(\theta)\ket{00}+\sin(\theta)\ket{11}.\label{eq:part_ent}
\end{equation}
When the state is not maximally entangled Alice and Bob can only establish perfect correlation when they both measure in the $Z$ basis. In this case, we thus assume that the key is only extracted from this measurement and set $K_B=1$. The entropy of the key-generating outcome is then given by
\begin{equation}
H_{1,1}(B)=h\left(\cos^2(\theta)\right)
\end{equation}
for the ideal implementation. The conditional entropy term $H_{\eta,v}(B|A)$ is always larger than (equal for optimally aligned measurements)
\begin{multline}\label{eq:HBA}
    H_{\eta,v}(B|A) 
= \eta(1-\eta)+h(\eta) + \eta^2\, h\left( \frac{1+v^2}{2}\right),
\end{multline} where $h(x)$ is the binary entropy.  

Plugging these expressions into Eq.~\eqref{eq: DIQKD bound}  gives
\begin{equation}\label{eq: DIQKD bound partial}
    r_\infty \leq  h\big(\cos^2(\theta)\big) (2\eta v -p_A-p_B-\eta^2v^2) -H_{\eta,v}(B|A).
\end{equation}
One notices that $h(\cos^2(\theta))$ is the only term depending on the parameter $\theta$, and it is straightforward to see that the optimal protocol (where the key rate bound is the less stringent) consists of preparing the maximally entangled state. We can thus set $\theta=\frac{\pi}{4}$ for which $H_{1,1}(B)=1$. To find the best bound it remains to maximize $p_A+p_B$ with the guidance of Table~\ref{tab:bounds} and the functions $ p_{K,N,d}'(\eta,v),  p_{N}''(\eta,v)$ or $  p'''(\eta,v)$ defined in Eqs.~(\ref{eq: many d} ,\ref{eq: p N qutits}, \ref{eq: all qubit PVMS p}).  To do so we consider three situations.\\

\noindent\textit {(i)} $(N_A,N_B,K_B) = (3,2,1)$ is the setting of the DIQKD CHSH protocol introduced in \cite{acin2006efficient,Acin2007}. Here the optimal bound reads
\begin{equation}
     r_\infty \leq  \, \eta v\big(2-p_{3}''(\eta,v)-p_{2}''(\eta,v)\big) -\eta^2v^2 -H_{\eta,v}(B|A). \label{eq:DIQKDmax B1}
\end{equation}
\\

\noindent \textit{(ii)} $(N_A,N_B,K_B) = (\infty,\infty,1)$ corresponds to the setting where only one measurement is used for key generation, but the number of test measurements is not restricted. Here the bounds 
\begin{equation}
     r_\infty \leq  \, \eta v\big(2-p_{1,\infty,2}'(\eta,v)-p'''(\eta,v)\big) -\eta^2v^2 -H_{\eta,v}(B|A), \label{eq:DIQKDmax B2}
\end{equation}
are to be used depending on the values of $\eta$ and $v$.
\\

\noindent\textit{(iii)} $(N_A,N_B,K_B) = (\infty,\infty,\infty)$ is the scenario with no restriction on the number of measurements. This is a meaningful scenario as for the maximally entangled state Alice and Bob can have perfect correlation for any number of measurements. Here the bound 
\begin{equation}
  r_\infty \leq  \, \eta v\big(2-2\,p'''(\eta,v)\big) -\eta^2v^2 -H_{\eta,v}(B|A) \label{eq:DIQKDmax B3}
\end{equation}
is optimal.\\

In Fig.~\ref{fig:qubit max} we plot the lines in the $(\eta,v)$ plane below which the bounds of Eqs.~(\ref{eq:DIQKDmax B1}-\ref{eq:DIQKDmax B3})
guarantee a zero key rate. The threshold values of $\eta$ and $v$ (for $v,\eta=1$ respectively) are reported in Table.~\ref{tab:DIQKD qubits}
\begin{figure}[t]
    \centering
    \resizebox{.95\linewidth}{!}{\begin{tikzpicture}
    \begin{axis}[
        myplotset,
        xlabel = {$v$},
        ylabel = {$\eta$},
        xtick={0.8706,0.8979,1},
        extra x ticks={0.8876}, 
        extra x tick style={tick label style={yshift=-10pt}},
        ytick={0.853,0.8828,1},
        extra y ticks={0.8744}, 
        extra y tick style={tick label style={yshift=-2pt}},
        legend pos =  south west,   
        legend cell align = left
      ]
    
      \addplot[green!70!black,thick,dashdotted] table {321nobin.txt};
      \addlegendentry{$(N_A,N_B,K_B)=(3,2,1)$};
      \addplot[orange,thick] table {infinf1nobin.txt};
      \addlegendentry{$(N_A,N_B,K_B)=(\infty,\infty,1)$};
      \addplot[red,thick,dashed] table {infinfinfnobin.txt};
      \addlegendentry{$(N_A,N_B,K_B)=(\infty,\infty,\infty)$};
    \end{axis}
    \end{tikzpicture} } 
    \caption{For the parameters $(\eta,v)$ below the lines the key rate $r_\infty$ is guaranteed to be zero. The lines show when the right-hand sides of the corresponding equation fall to zero. (From top to bottom) for Eq.~\eqref{eq:DIQKDmax B1} : \textit{(i)} $(N_A,N_B,K_B)=(3,2,1)$ dash-dotted line 
    , Eq.~\eqref{eq:DIQKDmax B2} : \textit{(ii)}  $(N_A,N_B,K_B)=(\infty,\infty,1)$ full line, and  Eq.~\eqref{eq:DIQKDmax B3} : \textit{(iii)} $(N_A,N_B,K_B)=(\infty,\infty,\infty)$  dashed line. This applies to DIQKD protocols based on two-qubit states, dichotomic measurements (ideally),  no-binning of the no-click outcome, and one-way public communication
    from Bob to Alice.}
    \label{fig:qubit max}
\end{figure}
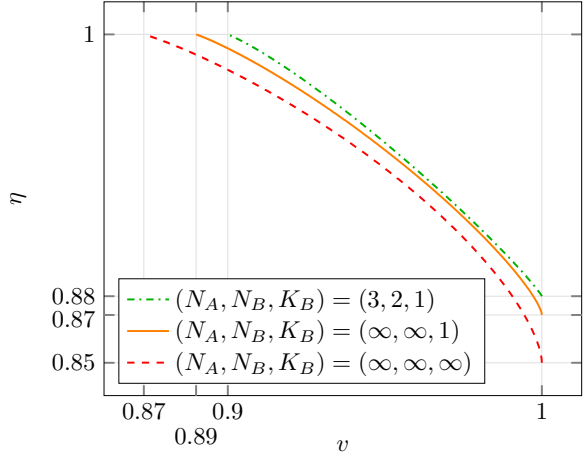
\begin{table}[t]
\centering
\begin{tabular}{|c | c |c |}
\hline
Setting& Threshold $\eta$  &  Threshold $v$ \\
$(N_A,N_B,K_B)$   & (at $ v=1 $) &  (at $ \eta =1 $) \\
\hline \hline
\rowcolor{lime}
(3,2,1) & 88.3 \% & 89.8 \% \\
\rowcolor{yellow}
\hline
($\infty$,$\infty$,1) &  87.4 \%&  88.8 \% \\
\rowcolor{pink}
\hline
($\infty$,$\infty$,$\infty$) & 85.3 \%  & 87.1 \% \\
\hline
\end{tabular}
\caption{Threshold efficiency $\eta$ and visibility $v$ for varions number of measurements $(N_A,N_B, K_B)$.  This applies to DIQKD protocols with the maximally entangled two-qubit state, dichotomic measurements (ideally), no-binning of the no-click outcome, and one-way public communication
from Bob to Alice.}\label{tab:DIQKD qubits}
\end{table}

\subsubsection{Application to qubit protocols with binning}\label{sec:binning}

A common strategy in the DIQKD protocols is to bin non-click outcomes of Bob with the most likely among the other two outcomes. After such binning the key rate bound in the Eq.~\eqref{eq: DIQKD bound partial} remains valid upon replacing the conditional entropy term with
\begin{multline}
    H_{\eta,v}^\theta(B|A) =  (1-\eta)\, h\left(\frac{\eta}{2}(1-v\cos(2\theta))\right)
    \\
    +\frac{\eta}{2}\left(1-v\cos(2\theta)\right)\, h\left(\eta\left(1-\frac{1-v^2}{2(1-v\cos(2\theta))}\right)\right)\\
    +\frac{\eta}{2}\left(1+v\cos(2\theta)\right)\, h\left(\eta\frac{1-v^2}{2(1+v\cos(2\theta))}\right).
\end{multline}

\begin{figure}[t]
    \centering
    \resizebox{.95\linewidth}{!}{\begin{tikzpicture}
    \begin{axis}[
        myplotset,
        xlabel = {$v$},
        ylabel = {$\eta$},
        xtick={0.8706,0.8979,1},
        extra x ticks={0.8876}, 
        extra x tick style={tick label style={yshift=-10pt}},
        ytick={0.6825,0.7422,1},
        extra y ticks={0.7268}, 
        extra y tick style={tick label style={yshift=-2pt}},
        legend pos =  south west,   
        legend cell align = left
      ]
    
      \addplot[green!70!black,thick,dashdotted] table {321bin.txt};
      \addlegendentry{$(N_A,N_B,K_B)=(3,2,1)$};
      \addplot[orange,thick] table {infinf1bin.txt};
      \addlegendentry{$(N_A,N_B,K_B)=(\infty,\infty,1)$};
      \addplot[red,thick,dashed] table {infinfinfbin.txt};
      \addlegendentry{$(N_A,N_B,K_B)=(\infty,\infty,\infty)$};
    \end{axis}
    \end{tikzpicture} } 
    \caption{For the parameters $(\eta,v)$ below the lines the key rate $r_\infty$ was found to be zero after a maximization over $\theta$. The lines show when the right-hand sides of the corresponding equation fall to zero. (From top to bottom) for Eq.~\eqref{eq:DIQKDmax B4} : \textit{(i')} $(N_A,N_B,K_B)=(3,2,1)$ dash-dotted line 
    , and Eq.~\eqref{eq:DIQKDmax B5} : \textit{(ii')}  $(N_A,N_B,K_B)=(\infty,\infty,1)$ full line.  For the curve, Eq.~\eqref{eq:DIQKDmax B6} : \textit{(iii')} $(N_A,N_B,K_B)=(\infty,\infty,\infty)$ dashed line, no maximization over $\theta$ is needed as it only makes sense to consider the maximally entangled state $\theta=\frac{\pi}{4}$. This applies to DIQKD protocols based on two-qubit states, dichotomic measurements (ideally),  binning of the no-click outcomes, and one-way public communication
    from Bob to Alice.}
    \label{fig:bin}
\end{figure}
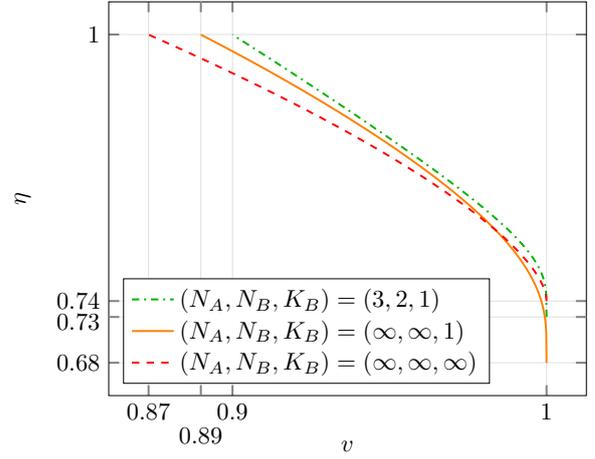

\begin{table}[t]
\centering
\begin{tabular}{|c | c |c |}
\hline
Setting& Threshold $\eta$  &  Threshold $v$ \\
$(N_A,N_B,K_B)$   & (at $ v=1 $) &  (at $ \eta =1 $) \\
\hline \hline
\rowcolor{lime}
(3,2,1) & 72.7 \% & 89.8 \% \\
\rowcolor{yellow}
\hline
($\infty$,$\infty$,1) &  68.3 \%&  88.8 \% \\
\rowcolor{pink}
\hline
($\infty$,$\infty$,$\infty$) & 74.2 \%  & 87.1 \% \\
\hline
\end{tabular}
\caption{Threshold efficiency $\eta$ and visibility $v$ for different number of measurements $(N_A,N_B, K_B)$.  This applies to DIQKD protocols with the partially entangled two-qubit state, dichotomic measurements (ideally), binning of the no-click outcomes, and one-way public communication
from Bob to Alice.}\label{tab:DIQKD qubits bin}
\end{table}

Note that in general this reduces Alice's entropy, i.e., $  H_{\eta,v}^\theta(B|A)$ here above is lower or equal to $H_{\eta,v}(B|A)$ in Eq. (\ref{eq:HBA}) since the binning channel applied by Bob is two-to-one. It follows that the key rate bounds that we obtain using binning are less stringent than without. With the new expression for  $H_{\eta,v}^\theta(B|A)$ let us consider the following settings. \\

\noindent \textit{(i')} For $(N_A,N_B,K_B)=(3,2,1)$  the optimal bound reads
\begin{multline}
   r_\infty \leq  h\big(\cos^2(\theta)\big) \big( \eta v(2-p_{3}''(\eta,v)-p_{2}''(\eta,v))-\eta^2 v^2\big) \\
   -H_{\eta,v}^\theta(B|A).
 \label{eq:DIQKDmax B4}  
\end{multline}

\noindent \textit{(ii')} For $(N_A,N_B,K_B)=(\infty,\infty,1)$ the optimal bound reads
\begin{multline}
 r_\infty \leq h\big(\cos^2(\theta)\big) \big( \eta v(2-p_{1,\infty,2}'(\eta,v)-p'''(\eta,v))-\eta^2 v^2\big)\\
  -H_{\eta,v}^\theta(B|A) \label{eq:DIQKDmax B5}    
\end{multline}

\noindent\textit{(iii')} For $(N_A,N_B,K_B) = (\infty,\infty,\infty)$ we only consider the case of the maximally entangled state $\theta=\frac{\pi}{4}$, as it is the only one where Alice and Bob can have perfect correlations in more than one measurement basis. We then get the optimal bound
\begin{equation}
  r_\infty \leq  \eta v\big(2-2\,p'''(\eta,v)\big) -\eta^2v^2
  -H^{\frac{\pi}{4}}_{\eta,v}(B|A) \label{eq:DIQKDmax B6}.
\end{equation}

To derive state-independent upper-bounds one needs to maximize the right-hand side of the above expressions with respect to $\theta$. Here, both quantities $H_{1,1}(B)$ and $H_{\eta,v}^{\theta}(B|A)$ have a nontrivial dependence on the angle $\theta$, and this maximization has to be done for all $\eta$ and $v$. In Fig.~\ref{fig:bin}, we display the thresholds that we found after maximizing the key rate over the angle $\theta$ in a heuristical way. Specifically, we used the Nelder–Mead technique \cite{nelder-mead}. In Table.~\ref{tab:DIQKD qubits bin} we give the corresponding threshold values.

\section{Discussion}\label{sec:discussion}

We have developed a general framework for characterizing the impact of losses in quantum communication setups with untrusted measurement devices. A key element in our work is the concept of a CMU, which should be viewed as a set of effective measurements, characterizing both the imperfections of the channel and the subsequent measurement apparatus. We investigated the joint measurability of this set, establishing a direct connection between this fundamental notion and quantum communication, in particular QKD.

Our work motivates the investigation of the (in)compatibility of sets of POVM under the combined effect of noise and loss, of which we discussed several cases. In particular, we have derived a sufficient condition for the compatibility of sets of $N$ unbiased qubit measurements (Lemma~\ref{Lemma:N qubit}). We have also seen that results on channel extendibility are very useful in this context, as they imply compatibility of CMUs solely based on the noise channel and the number of measurements that compose it. We derived a general result on the extendibility of concatenated noise channels.
Our results can be applied to any protocol or experimental scenario with untrusted measurement devices, including DI and SDI protocols, providing insight into the admissible trade-off between noise and loss.  Obtaining more results in this direction would be desirable.

In the context of QKD, the notion of joint measurability of a CMU leads to necessary criteria for extracting a secure key. The analysis of QKD protocols also led us to introduce the concept of partial joint measurability. We showed that this notion differs from standard joint measurability, in the sense that there exist sets of measurements that are partially jointly measurable but not (fully) jointly measurable. It would be interesting to investigate this concept further, in particular how to cast partial joint measurability as a semi-definite program. We note that this concept as we defined it differs from the notions discussed in previous works, where sets of POVMs that are incompatible can feature subsets of POVMs that are jointly measurable \cite{Heinosaari2008,Liang20111,uola2023}. 
In contrast, our definition is closely related to the notion of no-exclusivity introduced recently for quantum instruments~\cite{Buscemi2023unifyingdifferent}, and to the notion of weak-compatibility introduced earlier in the same context~\cite{DAriano_2022}. In particular, the notion of partial joint-measurability that we introduced for a channel followed by a set of POVMs can be generalized to an analogous notion of partial no-exclusivity for a channel followed by a set of instruments. Lemma~\ref{Lemma:k-ext} then can naturally be extended to this setting.

Finally, as an outlook note that the intuition behind the notion of partial joint measurability can be used to improve the class of attacks on DIQKD considered in \cite{farkas2021bell,lukanowski2022upper}. These attacks are based on decompositions of the correlations observed by Alice and Bob $p(a,b|x,y)= q\, p_L(a,b|x,y)+ (1-q) p_{NL}(a,b|x,y)$  into a mixture of local $p_L(a,b|x,y)$ and non-local $p_{NL}(a,b|x,y)$ correlations. In fact, the locality of correlation $p_L(a,b|x,y)$ is sufficient but not necessary for Eve to predict the key. Since she only needs to predict the outcomes of the key-generating measurements, we introduced here attacks where $p(a,b|x,y)$ is decomposed into a mixture of non-local quantum correlations and partly-\JM{} $p_{JM}(a,b|x,y)$ correlations. The partly-\JM{} correlations can be associated with a quantum-classical state, where only the classical register is used to simulate the outputs of the key-generating measurements, whereas for the other measurements the quantum register is used, hence can display non-locality. 

\section*{Acknowledgments}
We thank Nicolas Cerf for fruitful discussions.
P.S., M.I., and N.B. acknowledge funding from the Swiss National Science Foundation (project 192244) and by the Swiss Secretariat for Education, Research and Innovation (SERI) under contract number UeM019-3.  M.M. and S.P. acknowledge funding from the VERIqTAS project within the QuantERA II Programme that has received funding from the European Union’s Horizon 2020 research and innovation programme under Grant Agreement No 101017733 and the F.R.S-FNRS Pint-Multi programme under Grant Agreement R.8014.21, from the European Union’s Horizon Europe research and innovation programme under the project ``Quantum Security Networks Partnership'' (QSNP, grant agreement No 101114043),  from the F.R.S-FNRS through the PDR T.0171.22, from the FWO and F.R.S.-FNRS under the Excellence of Science (EOS) programme project 40007526, from the FWO through the BeQuNet SBO project S008323N, from the Belgian Federal Science Policy through the contract RT/22/BE-QCI and the EU ``BE-QCI'' programme.
S.P. is a Research Director of the Fonds de la Recherche Scientifique – FNRS.

Funded by the European Union. Views and opinions expressed are however those of the authors only and do not necessarily reflect those of the European Union. The European Union cannot be held responsible for them.

\appendix

\section{Proof of the bound (\ref{eq:lb-guassian}) in the case of homodyne measurements}\label{ap:har}
Homodyne measurements corresponding to the observables
\begin{equation}
    \hat{X}(\theta)=\cos\theta \hat{X}+\sin\theta \hat{P},
\end{equation}
specified by the angle $\theta$ and where $\hat{X}$ and $\hat{P}$ are the position and momentum quadrature operators
\begin{equation}
    \hat{X}=\frac{\hat{a}+\hat{a}^\dagger}{\sqrt{2}},\quad\text{and}\quad \hat{P}=\frac{\hat{a}-\hat{a}^\dagger}{i\sqrt{2}}.
\end{equation}
We now show that when the measurements $M_y$ in a thermal-noise CMU correspond to the above observables (with the angle $\theta(y)$ depending on $y$), the resulting effective POVMs are \JM{} independently of the total number $N$ of such measurements. 

In a thermal-noise model, the annihilation operators follow the input-output relation
    \begin{equation}
        \hat{a}_\text{out}=\sqrt{\eta}\hat{a}_\text{in}+\sqrt{1-\eta}a_{\tau_\nu},
    \end{equation}
where $a_{\tau_\nu}$ is the annihilation operator acting on the thermal state that is combined with the input state in our noise model. We can use the latter relation to express the effective CMU quadratures that are implemented by the CMU as
\begin{equation}\label{eq:eff-quad}
    \hat{X}_\text{out}(\theta)=\sqrt{\eta}\hat{X}_\text{in}(\theta)+\sqrt{1-\eta}\hat{X}_{\tau_\nu}(\theta).
\end{equation}    
    
We now introduce a single parent POVM that simulates the results of the measurements of such quadratures for any value of $\theta$. This parent POVM consists in performing first a heterodyne measurement on the input state, i.e, we split the initial mode into two modes using a $50/50$ beam splitter and measuring the $X$ quadrature of one output mode and the $P$ quadrature of the other output mode. The obtained classical outcomes, $x$ and $p$, are then sent to the measurement device, where given the input angle $\theta$, it outputs 
\begin{equation}
    x_\theta=G(\cos\theta x+\sin\theta p)+\Delta,\label{eq:output}
\end{equation}
where $G$ is a constant real number, while $\Delta$ is a random real number with centered Gaussian distribution of variance $\sigma^2$. 

In this equivalent CMU, the input-output relations for the annihilation operators after the $50/50$ beam splitter are
\begin{equation}
    \hat{a}_1=\frac{\hat{a}_\text{in}+\hat{v}}{\sqrt{2}},\quad\text{and}\quad \hat{a}_2=\frac{\hat{a}_\text{in}-\hat{v}}{\sqrt{2}},
\end{equation}
where $\hat{v}$ is the annihilation operator of the vacuum state entering the dark port of our beam splitter and $\hat{a}_1$ and $\hat{a}_2$ are the annihilation operators of the two output modes. The quadrature $\hat{X}$ measured in the first output mode and the quadrature $\hat{P}$ measured in the second output mode can then be written as
\begin{equation}
    \hat{X}_1=\frac{\hat{X}_\text{in}+\hat{X}_v}{\sqrt{2}},\quad\text{and}\quad \hat{P}_2=\frac{\hat{P}_\text{in}-\hat{P}_v}{\sqrt{2}}.
\end{equation}

The probability distribution of the final output $x_\theta$ of the CMU after the post-processing \eqref{eq:output} is then fully characterized by the moments $\langle x_\theta^n\rangle$. 
This strategy simulates the original CMU if these moments coincide with the moments $ \langle \hat{X}(\theta)^n \rangle_\text{th}$ of the distribution of the effective observables (\ref{eq:eff-quad}). We now determine such moments.

We first compute the target moments.
\begin{align}
    &\langle \hat{X}(\theta)^n \rangle_\text{th} = \left\langle \left(\sqrt{\eta}\hat{X}_\text{in}\theta)+\sqrt{1-\eta}\hat{X}_{\tau_\nu}(\theta)\right)^n \right\rangle \\
    &=\sum_{k=0}^n\binom{n}{k}\sqrt{\eta}^{n-k}\langle \hat{X}_\text{in}(\theta)^{n-k}\rangle  \sqrt{1-\eta}^k\langle \hat{X}_{\tau_\nu}(\theta)^k\rangle \\
    &=\sum_{k=0}^{\lfloor n/2 \rfloor}\binom{n}{2k}\sqrt{\eta}^{n-2k}\langle \hat{X}_\text{in}(\theta)^{n-2k}\rangle \nonumber \\
    &\pushright{(1-\eta)^k\langle \hat{X}_{\tau_\nu}(\theta)^{2k}\rangle} \\
    &=\sum_{k=0}^{\lfloor n/2 \rfloor}\binom{n}{2k}\sqrt{\eta}^{n-2k}\langle \hat{X}_\text{in}(\theta)^{n-2k}\rangle \nonumber \\
    &\pushright{(2k-1)!!\left(\frac{1}{2}(1-\eta+\epsilon\eta)\right)^k,} \label{eq:mom-noise}
   \end{align} where we used the fact that the odd moments of the thermal state in Eq.~(\ref{eq:ths}) are zero, while the even ones are 
   \begin{equation}
   (2k-1)!!\left(\frac{1}{2}+\frac{\epsilon\eta}{2(1-\eta)}\right)^k.
\end{equation}
   
We now determine the moments following from the \JM{} strategy.
\begin{align}
\langle x_\theta^n\rangle &=\langle \left(G(\cos\theta \hat{X}_1+\sin\theta \hat{P}_2)+\Delta\right)^n \rangle \\
&=\left\langle \left(\frac{G}{\sqrt{2}}( \hat{X}_\text{in}(\theta)+ \hat{X}_v(-\theta))+\Delta\right)^n \right\rangle \\
&=\sum_{k=0}^n \binom{n}{k}\left(\frac{G}{\sqrt{2}}\right)^{n-k}\langle X_\text{in} (\theta)^{n-k}\rangle \nonumber \\ 
&\pushright{\left\langle\left(\frac{G}{\sqrt{2}} \hat{X}_v(-\theta)+\Delta\right)^k \right\rangle} \\
&=\sum_{k=0}^{\lfloor n/2\rfloor} \binom{n}{2k}\left(\frac{G}{\sqrt{2}}\right)^{n-2k}\langle X_\text{in} (\theta)^{n-2k}\rangle \nonumber \\ 
&\pushright{\left\langle\left(\frac{G}{\sqrt{2}} \hat{X}_v(-\theta)+\Delta\right)^{2k} \right\rangle} ,
\end{align} where in the last line we used the fact that the sum of two random variables with Gaussian central distribution has zero odd moments. Now, we have
\begin{align}
&\left\langle \left(\frac{G}{\sqrt{2}} \hat{X}_v(-\theta)+\Delta\right)^{2k} \right\rangle = \nonumber \\
&=\sum_{l=0}^{2k}\binom{2k}{l}\left(\frac{G}{\sqrt{2}}\right)^{2k-l}\langle \hat{X}_v(-\theta)^{2k-l}\rangle\langle\Delta^l\rangle \\
&=\sum_{l=0}^{k}\binom{2k}{2l}\left(\frac{G}{\sqrt{2}}\right)^{2(k-l)} \left(\frac{1}{2}\right)^{k-l}\sigma^{2l}\\
&\pushright{(2(k-l)-1)!!(2l-1)!!},
\end{align} where we used again that the odd moments are zero and that the even ones are
\begin{equation}
\langle \hat{X}_v(-\theta)^{2(k-l)}\rangle=\left(\frac{1}{2}\right)^{k-l}(2(k-l)-1)!!,
\end{equation}
\begin{equation}
\langle\Delta^l\rangle=\sigma^{2l} (2l-1)!!.
\end{equation} Using $(2n)!=(2n)!!(2n-1)!!$ and $(2n)!!=2^n n!$, we get
\begin{equation}
\binom{2k}{2l}(2(k-l)-1)!!(2l-1)!!=\binom{k}{l}(2k-1)!!,
\end{equation} and thus
\begin{equation}
\left\langle \left(\frac{G}{\sqrt{2}} \hat{X}_v(-\theta)+\Delta\right)^{2k} \right\rangle = (2k-1)!!\left(\frac{G^2}{4}+\sigma^2\right)^k.
\end{equation} We finally obtain that
\begin{multline}
\langle x_\theta^n\rangle =\sum_{k=0}^{\lfloor n/2\rfloor} \binom{n}{2k}\left(\frac{G}{\sqrt{2}}\right)^{n-2k}\langle X_\text{in} (\theta)^{n-2k}\rangle  \\ 
(2k-1)!!\left(\frac{G^2}{4}+\sigma^2\right)^k.
\end{multline} 
   
It thus follows that $\langle \hat{X}(\theta)^n \rangle_\text{th} = \langle x_\theta^n\rangle$ for all $n$ if 
\begin{equation}
    G=\sqrt{2\eta},\quad\text{and}\quad\sigma^2=\frac{1}{2}(1-2\eta+\epsilon\eta).
\end{equation}
Since $\sigma^2\geq 0$, we conclude that this strategy works as long as (\ref{eq:lb-guassian}) holds.

\section{Proof of Lemma~\ref{Lemma:gauss-decomp}}
\label{app: gaussian decomp}
In this section, we prove Lemma~\ref{Lemma:gauss-decomp} in the main text. But first, we recall some elements of the formalism of Gaussian channels, states, and their phase-space representation.

The Wigner function  of a single-mode Gaussian state with the covariance matrix $V$ and displacement vector $\bm \mu$ is given by
\begin{equation}
 W(\alpha)=   \frac{1}{\sqrt{\text{det}(\pi V/2)}}\exp\left(-2(\bm \alpha -\bm \mu)^T V^{-1}(\bm \alpha - \bm \mu)\right).
\end{equation}
with $\bm \alpha = \binom{\text{Re} \, \alpha}{\text{Im} \, \alpha}$. By definition, the pair $V,\bm \mu$ thus uniquely specifies a Gaussian state. In particular, for  a coherent state $\ket{\gamma}$ we have $\bm \mu = \bm \gamma=\binom{\text{Re} \, \gamma}{\text{Im} \, \gamma}$ and $V= \mathbbm{1}$.

A single mode Gaussian channel $\Phi$ maps Gaussian states to Gaussian states. It can be represented by matrices $X,Y$ and a vector $\bm \delta$, which transform the Wigner function as follows~\cite{serafini2017quantum,lami2019}
\begin{align}
\bm \mu   \mapsto X \bm \mu + \bm \delta\\
V \mapsto X V X^T + Y
\end{align}
Therefore, if a Gaussian channel acts on a coherent state, the resulting Wigner function of state $\Phi(\ketbra{\gamma})$ is Gaussian and characterized by the pair $V'=  XX^T+Y$ and $\bm \mu' = X\bm \gamma+\bm \delta$. 

Finally, let us compute the Q-function corresponding to this state. It is well known that it is the convolution of the Wigner function with a Gaussian
\begin{equation}
    Q(\alpha) = \frac{2}{\pi} \int \dd^2 \beta\,  W(\beta)e^{-2|\alpha-\beta|^2}.
\end{equation}
Hence the Q-function is also a Gaussian with the displacement vector $\bm \mu_Q=\bm \mu'$ and covariance matrix $V_Q = V' + \mathbbm{1}$ (note that a convolution of two distributions describes a random variable given by the sum of the random variables described by the distributions). Thus if a Gaussian channel $(X,Y,\delta)$ acts on a coherent state $\ket{\gamma}$, the Q-function of the finals state is Gaussian with 
\begin{align}\label{eq: mu gauss}
    \bm \mu_Q &= X \bm \gamma +\bm \delta \\
    \label{eq: V gauss}
     V_Q &= X X^T +Y +\mathbbm{1} .
\end{align}
Now let us come back to the Lemma that we want to prove.

\setcounter{lem}{7}
\begin{lem}
    A thermal-noise channel $\Phi_{\eta,\epsilon}$ that is not $N$-extendable cannot be decomposed as 
\begin{equation}\label{eq: gauss dec ap}
    \Phi_{\eta,\epsilon} =p\, \Phi_G + (1-p) \Phi' \qquad (0<p<1)
\end{equation}
where $\Phi_G$ is a $N$-extendable Gaussian channel and $\Phi'$ is an arbitrary channel.
\end{lem} 
\begin{proof}
To prove this result we use the following Lemma, which is demonstrated below.

\begin{lem}\label{Lemma:append}
The thermal-noise channel $\Phi_{\eta,\epsilon}$ can admit a decomposition
\begin{equation}\label{eq: gauss dec2}
    \Phi_{\eta,\epsilon} =p\, \Phi_G + (1-p) \Phi'
\end{equation}
with a Gaussian channel $\Phi_G$ characterized by the matrices $X$ and $Y$ (see e.g.\cite{serafini2017quantum}), a positive trace preserving map $\Phi'$, and
some $p>0$, only if 
\begin{align}
    X=\sqrt{\eta}\mathbbm{1} \qquad Y \leq (1-\eta+\eta \varepsilon)\mathbbm{1}.
\end{align}
\end{lem}

Let us now take a thermal channel $\Phi_{\eta,\epsilon}$ which is not $N$-extendable, that is with 
\begin{equation}\label{eq: N-ext gauss}
\eta (1-\epsilon/2)> \frac{1}{N}
\end{equation}
accordingly to \cite{lami2019} (note that on the level of the channel, this is a necessary and sufficient condition). Consider any Gaussian channel $\Phi_G$ appearing in the decomposition of Eq.~\eqref{eq: gauss dec2}. On the one hand we know that it is $N$-extendable if and only if~\cite{lami2019}
\begin{equation}\label{eq: ext cond}
\sqrt{\det Y} \geq 1-\frac{1}{N} +\left|\det X -\frac{1}{N}\right|.
\end{equation}
On the other, Lemma 1 implies that $\sqrt{\det Y} \leq (1-\eta+\eta \epsilon)$ and $\det X =\eta$ must hold. By using these identities together with the bound in Eq.~\eqref{eq: N-ext gauss}, it is easy to see that the condition of Eq.~\eqref{eq: ext cond} can never be satisfied. Hence, the decomposition in Eq.~\eqref{eq: gauss dec2} is impossible.
\end{proof}

\begin{proof}[Proof of Lemma~\ref{Lemma:append}]
First note that the case $\eta=0$ is trivial, we thus assume $\eta > 0$.

Consider the map $\tilde \Phi = (1-p)\Phi' = \Phi_{\eta,\epsilon} - p\Phi_G$, which must be at least positive for the decomposition to make physical sense. Let us act with $\tilde\Phi$ on a coherent state $\ket{\gamma}$, the Q-function of final state $\tilde \Phi(\ketbra{\gamma})$ is given by
\begin{equation}
    \tilde Q_\gamma(\alpha) = Q^{\eta,\epsilon}_\gamma(\alpha) - p\,  Q^G_{\gamma}(\alpha),
\end{equation}
where $Q^{\eta,\epsilon}_\gamma$ and $Q^G_{\gamma}$ are the Q-functions associated to the states $\Phi_{\eta,\epsilon}(\ketbra{\gamma})$ and $\Phi_G(\ketbra{\gamma})$. The Q-function of a state is proportional to the probability density of the output of the heterodyne measurement. Therefore, if $\tilde \Phi$ is a positive channel, the function $\tilde Q_\gamma(\alpha)$ must be positive for all $\alpha$ and $\gamma$, i.e. we must have 
\begin{equation}
Q^{\eta,\epsilon}_\gamma(\alpha) \geq p\,  Q^G_{\gamma}(\alpha)
\end{equation}
for some nonzero $p$ and all $\alpha$ and $\gamma$. This is equivalent to requiring that the ratio $\frac{Q^G_{\gamma}(\alpha)}{Q^{\eta,\epsilon}_\gamma(\alpha)}$
is bounded. Since both distributions are Gaussian this is equivalent to the following bound 
\begin{equation}\label{eq: condition}
    (\bm \alpha -\bm \mu_{th})^T V_{th}^{-1} (\bm \alpha -\bm \mu_{th})^T -(\bm \alpha -\bm \mu_G)^T V_{G}^{-1} (\bm \alpha -\bm \mu_G)<\infty,
\end{equation}
where for the thermal channel $\Phi_{\eta,\epsilon}$ we have $\bm \mu_{th} =\sqrt{\eta} \bm \gamma$ and $V_{th}= (2+\epsilon\eta)\mathbbm{1}$, while for a general Gaussian channel $V_G= XX^T+Y+\mathbbm{1}$ and $\bm \mu_G =X \bm \gamma +\bm \delta$ accordingly to Eqs.~(\ref{eq: mu gauss},\ref{eq: V gauss}). This condition is only nontrivial in the limit where $|\alpha|,|\gamma|$ or both go to infinity. In this limit, $\bm \delta$ plays no role, and we can thus ignore it.

Plugging the values of $\bm \mu_{th}$, $V_{th}$ and $\bm \mu_G$ in Eq.~\eqref{eq: condition} we obtain
\begin{equation}
\frac{|\alpha - \sqrt{\eta}\gamma|^2}{2+\eta\epsilon}
- (\bm \alpha -X \bm \gamma)^T (V_G)^{-1} (\bm \alpha -X \bm \gamma)\leq \infty.
\end{equation}
This is a bilinear form in the varibles $\bm \xi =(\text{Re} \alpha,\text{Im} \alpha, \text{Re} \gamma, \text{Im} \gamma)^T$ and can be expressed as
\begin{align}    
\bm \xi^T M \bm \xi &\leq \infty   \\
M &=
\left(\begin{array}{c|c}
\frac{\mathbbm{1}}{2+\eta \epsilon}- V_G^{-1} & -\frac{\sqrt{\eta} \mathbbm{1}}{2+\eta \epsilon}+ X^T V_G^{-1} \\
\hline
-\frac{\sqrt{\eta} \mathbbm{1}}{2+\eta \epsilon}+ V_G^{-1} X & \frac{\eta \mathbbm{1}}{2+\eta \epsilon}+ X^T V_G^{-1} X
\end{array}
\right),
\end{align}
or simply $M\leq 0$. The positivity of this matrix $M$  remains unchanged if we multiply it by the diagonal matrix 
\begin{equation}
\text{diag}(\sqrt{2+\eta\epsilon},\sqrt{2+\eta\epsilon},-\frac{\sqrt{2+\eta\epsilon}}{\sqrt{\eta}}, -\frac{\sqrt{2+\eta\epsilon}}{\sqrt{\eta})})
\end{equation}
from the left and from the right. This allows us to rewrite the condition as 
\begin{equation}
  \left(\begin{array}{c|c}
\mathbbm{1}- W^{-1} & \mathbbm{1}- Z^T W^{-1} \\
\hline
\mathbbm{1}- W^{-1} Z & \mathbbm{1}- Z^T W^{-1} Z
\end{array}
\right)\leq 0  
\end{equation}
or
\begin{equation}\label{eq: D20}    
\mathbbm{1}_4 \leq \left(\begin{array}{c|c}
W^{-1} & Z^T W^{-1} \\
\hline
 W^{-1} Z & Z^T W^{-1} Z
\end{array}
\right)
\end{equation}
for $W^{-1}= (2+\eta \epsilon )V_G^{-1}$ and $Z = X/\sqrt{\eta}$.

We now consider two cases. First, assume that the matrix $X$ (and hence $Z$) is invertible. Multiplying the Eq.~\eqref{eq: D20} with $\text{diag}(\mathbbm{1}|Z^{-1})$ from the right and  $\text{diag}(\mathbbm{1}|(Z^{-1})^T)$ from the left does not change the positivity of a matrix, and allows us to  cast the inequality in the form 
\begin{equation}    
\left(\begin{array}{c|c}
\mathbbm{1} & (Z^{-1})^{T} \\
\hline
 Z^{-1} & (Z^{-1})^{T} Z
\end{array}
\right) \leq \left(\begin{array}{c|c}
W^{-1} & W^{-1} \\
\hline
 W^{-1}  &  W^{-1} 
\end{array}
\right).
\end{equation}
The right hand side is block-diagonal $\binom{1}{1}\binom{1}{1}^T \otimes W^{-1}$. And since the inequality must hold in both blocks, we get
\begin{align}
    (\mathbbm{1}+{Z^{-1}}^T)(\mathbbm{1}+Z^{-1}) &\leq 4 W^{-1}\\
    (\mathbbm{1}-{Z^{-1}}^T)(\mathbbm{1}-Z^{-1})&\leq 0.
\end{align}
The second inequality implies $\mathbbm{1}-Z^{-1}=0$, or $X=\sqrt{\eta}\mathbbm{1}$. This reduces the fist one to $W^{-1}\geq \mathbbm{1}$, or $W\leq \mathbbm{1}$. Recalling that $W= V_G/(2+\eta \epsilon)$ with $V_G = XX^T+Y+1=\eta \mathbbm{1}+Y+1$ we  obtain
\begin{equation}
    \frac{(1+\eta) \mathbbm{1} + Y}{2+\eta \epsilon} \leq \mathbbm{1},
\end{equation}
leading to
\begin{equation}
    Y \leq (1-\eta +\eta \epsilon) \mathbbm{1}.
\end{equation}
This proves the Lemma for an invertible $X$.

Next, consider the remaining case where $X$ is not invertible. There is thus a normalized vector $\bm v$ such that $X\bm v =Z\bm v = 0$. Multiplying the Eq~\eqref{eq: D20} with $(0,0|\bm v)$ from the right and $(0,0|\bm v)^T$ from the left we find
\begin{equation}
1\leq 0.
\end{equation}
$X$ must thus be invertible, which concludes the proof.
\end{proof}

Note that the Lemma is most probably also true with an if and only statement. The if direction could be shown by an explicit decomposition, noting that the thermal-noise channel can be viewed as a loss channel followed by a random displacement with a normally distributed complex amplitude.

\bibliography{references}

\end{document}